\documentclass[journal,twoside,10pt]{IEEEtran}
\usepackage{amsmath}
\usepackage{mathrsfs}

\usepackage{amsfonts}
\usepackage{cite}
\usepackage{graphicx,color,overpic}
\usepackage{amsmath}
\usepackage{times}
\usepackage{latexsym}
\usepackage{bm}
\usepackage{amssymb}
\usepackage{cases}
\usepackage{array}
\usepackage{fancyhdr}
\usepackage{setspace}
\usepackage{subfigure}
\usepackage{url}
\usepackage{algorithm}
\usepackage{algorithmic}
\usepackage{multirow}

\usepackage{citesort}
\usepackage{epsfig}
\usepackage{fancybox}
\usepackage{textcomp}
\usepackage{multirow}
\usepackage{setspace}
\usepackage{psfrag}
\usepackage{color}
\newtheorem{theorem}{Theorem}
\newtheorem{lemma}{Lemma}

    \def\Complex{{\rm\rule[.23ex]{.03em}{1.1ex}\kern-.3em{C}}}

    \newcommand{\be}{\begin{equation}} \newcommand{\ee}{\end{equation}}
    \newcommand{\bea}{\begin{eqnarray}} \newcommand{\eea}{\end{eqnarray}}
    \newcommand{\benum}{\begin{enumerate}} \newcommand{\eenum}{\end{enumerate}}

    \newcommand{\qa}{{\bf a}}
    \newcommand{\qb}{{\bf b}}

    \newcommand{\qe}{{\bf e}}
    \newcommand{\qf}{{\bf f}}
    
    \newcommand{\qh}{{\bf h}}

    \newcommand{\qn}{{\bf n}}

    \newcommand{\qs}{{\bf s}}

    \newcommand{\qv}{{\bf v}}
    \newcommand{\qw}{{\bf w}}
    \newcommand{\qx}{{\bf x}}
    \newcommand{\qy}{{\bf y}}
    \newcommand{\qz}{{\bf z}}

    \newcommand{\qA}{{\bf A}}
    \newcommand{\qB}{{\bf B}}
    \newcommand{\qC}{{\bf C}}

    \newcommand{\qF}{{\bf F}}
    
    \newcommand{\qH}{{\bf H}}
    \newcommand{\qI}{{\bf I}}
    \newcommand{\qJ}{{\bf J}}

    \newcommand{\qN}{{\bf N}}
    
    \newcommand{\qP}{{\bf P}}
    \newcommand{\qQ}{{\bf Q}}
    \newcommand{\qR}{{\bf R}}
    \newcommand{\qS}{{\bf S}}
    \newcommand{\qT}{{\bf T}}
    \newcommand{\qU}{{\bf U}}
    
    \newcommand{\qW}{{\bf W}}
    \newcommand{\qX}{{\bf X}}
    \newcommand{\qY}{{\bf Y}}

    \newcommand{\qPhi}{{\boldsymbol \Phi}}

     \newcommand{\qphi}{{\boldsymbol \phi}}

    \newcommand{\ty}{{\tilde{y}}}

    \newcommand{\tqN}{{\tilde{\qN}}}
    \newcommand{\tqY}{{\tilde{\qY}}}

    \newcommand{\bbE}{{\mathbb E}}

    \newcommand{\bbC}{{\mathbb C}}

    \newcommand{\calA}{{\mathcal A}}

    \newcommand{\calN}{{\mathcal N}}
    \newcommand{\calS}{{\mathcal S}}
    
    \newcommand{\calP}{{\mathcal P}}

    \newcommand{\diag}{{\sf diag}}

\begin{document}

\title{On the Downlink Average {Energy }Efficiency of Non-Stationary XL-MIMO}

\author{Jun Zhang,~\IEEEmembership{Senior Member,~IEEE,}
        Jiacheng Lu,~\IEEEmembership{Graduate Student Member,~IEEE,}
        Jingjing Zhang, \\
        Yu Han,~\IEEEmembership{Member,~IEEE,}
        Jue Wang,~\IEEEmembership{Member,~IEEE,}
        Shi Jin,~\IEEEmembership{Senior Member,~IEEE}
\thanks{J. Zhang, J. Lu, and J. J. Zhang are with Jiangsu Key Laboratory of Wireless Communications, Nanjing University of Posts and Telecommunications, Nanjing 210003, China, Email: \{zhangjun, 1020010209, 1020010323\}@njupt.edu.cn.}
\thanks{H. Yu and S. Jin are with the National Mobile Communications Research Laboratory, Southeast University, Nanjing 210096, China, Email: \{hanyu,jinshi\}@seu.edu.cn.}
\thanks{J. Wang is with School of Information Science and Technology, Nantong University, Nantong 226019, China, Email: wangjue@ntu.edu.cn.}
\thanks{This paper was presented in part at the 2022 WCSP \cite{Zhang22WCSP}.}
}

\maketitle
\begin{abstract}
Extra large-scale multiple-input multiple-output (XL-MIMO) is a key technology for future wireless communication systems. This paper considers the effects of visibility region (VR) at the base station (BS) in a non-stationary multi-user XL-MIMO scenario, where only partial antennas can receive users' signal. In time division duplexing (TDD) mode, we first estimate the VR at the BS by detecting the energy of the received signal during uplink training phase. The probabilities of two detection errors are derived and the uplink channel on the detected VR is estimated. In downlink data transmission, to avoid cumbersome Monte-Carlo trials, we derive a deterministic approximate expression for ergodic {average energy efficiency (EE)} with the regularized zero-forcing (RZF) precoding. In frequency division duplexing (FDD) mode, the VR is estimated in uplink training and then the channel information of detected VR is acquired from the feedback channel. In downlink data transmission, the approximation of ergodic average {EE} is also derived with the RZF precoding. Invoking approximate results, we propose an alternate optimization algorithm to design the detection threshold and the pilot length in both TDD and FDD modes. The numerical results reveal the impacts of VR estimation error on ergodic average {EE} and demonstrate the effectiveness of our proposed algorithm.
\end{abstract}

\section{Introduction}
Extra large-scale multiple-input multiple-output (XL-MIMO), featured by utilizing of an extremely large antenna array (ELAA) at base station (BS) and serving many users simultaneously, has become a fundamental technology to support ambitious visions of future sixth generation mobile communication systems (6G) \cite{Cui23MCN}. XL-MIMO provides higher array gain, better ability of spatial multiplexing, and remarkable improvement in spectral efficiency (SE), which helps significantly alleviate the short of spectral resources \cite{Chen20WC}. Besides, as compared to massive MIMO, preferable channel properties such as channel hardening and asymptotic orthogonality are better seen in XL-MIMO \cite{Feng22JIoT}.

However, the sharp increase on the array size also incurs new problems.
First, the complexity of channel estimation and signal processing becomes overwhelming as array dimension grows, which poses great challenges to both hardware and software design \cite{Feng23VTM}.
Second, the much larger Rayleigh distance of ELAA results in new channel characteristics \cite{Cui23MCN}. In traditional MIMO systems, the BS-user distance is usually beyond the Rayleigh distance, such that planar wavefront can be assumed when describing the wireless channel \cite{2003Antennas}. However, this is no longer true in XL-MIMO systems, where near-field propagation has to be considered.
Third, the spatial non-stationarity becomes non-negligible since the propagation between ELAA and users can be partially blocked by scatters or obstacles \cite{Jose22TWC,Han22JSTSP}. As a result, a user may merely see only partial antennas instead of the entire array, and these antennas that can be observed by the user are referred to its visibility region (VR) \cite{Carvalho20WC}. The existence of VR may reduce the array gain, on the other hand, it can also help reduce the complexity of precoding and user scheduling \cite{Jose22TWC,Lucas21EUSICPO,Nishimura20CL,Yang20TVT,Souza23TVT,Ali19WC}.

To cope with abovementioned problems, many researches investigate the channel modeling and VR features of XL-MIMO in recent years.
In \cite{Lu22CL}, the authors investigated the performances of maximal-ratio combining (MRC), zero-forcing (ZF), and minimum mean-square error (MMSE) beamforming under spherical wavefront, which provides additional freedom to suppress inter-user interference in the distance domain. In \cite{Lu22Tcom}, the authors further considered the size of physical array elements and established a unified model for different implementation forms of XL-MIMO, including discrete antenna array and continuous surface. The above works were conducted assuming line-of-sight (LoS) propagation between users and the ELAA. However, the propagation paths could be partially blocked by obstacles and thus VR effect occurs. To evaluate the impact of VR, the authors compared the signal-to-interference-plus-noise ratio (SINR) performance of conjugate beamforming (CB) and ZF precoders under both stationary and non-stationary channels in \cite{Ali19WC}, revealing that VR may bring favorable influences in reducing inter-user interferences.
Similarly, the performances of MRC and linear minimum mean squared error (LMMSE) receiver in the presence of VR were analyzed in \cite{Yang20TVT}. Considering a non-stationary propagation channel, the authors leveraged user grouping and plane-wave approximation to reduce the complexity of precoding in XL-MIMO, and developed two novel ZF precoders in \cite{Lucas21EUSICPO}.
Besides, precoding design and several user scheduling schemes were also elaborated by utilizing the independence of different users' VR in \cite{Jose22TWC,Nishimura20CL,Souza23TVT}. In \cite{Jose22TWC}, an VR XL-MIMO (NOVR-XL) protocol was proposed for joint user access and scheduling, which improved the sum-rate by taking advantages of the non-overlapping property of different users' VR. The authors further modified the classical strongest user collision resolution (SUCRe) protocol by exploring overlapping VR in \cite{Nishimura20CL}, and the modified protocol allowed more access attempts in XL-MIMO. Moreover, to ensure the quality-of-service (QoS) requirements of different users, a QoS-aware user scheduling scheme was proposed in \cite{Souza23TVT}, achieving fair coverage for a cell where LoS and non-LoS (NLoS) transmission coexist.

Note that abovementioned works all assumed perfect VR knowledge can be obtained at the BS. Despite VR can reduce the precoding complexity and enhance the efficiency of user scheduling, the channel estimation on VR can be challenging since VR distribution can be random and discontinuous, or successively in a block, or partially overlapped for different users. To deal with the VR estimation issue, a turbo orthogonal approximate message passing (Turbo-OAMP) algorithm was proposed in \cite{Zhu21TSP}, where the structured sparsity of XL-MIMO channel was discussed. By jointly detecting the activity of both users and subarrays, a bilinear Bayesian inference algorithm was proposed in \cite{Hiroki22TWC} to estimate the spatial non-stationary channel pattern. In \cite{Han22JSTSP,Tian23WCL}, the VR detection and user localization were jointly accomplished by sensing the flat region in the accumulating function of different antennas' energy. Moreover, a U-shaped dedicated multilayer perceptron (MLP) network was proposed in \cite{Xiao23WCL} to reconstruct the channel on VR with limited pilot overhead. By utilizing the dependence among users, clusters, and VR at the BS,
the authors proposed a novel VR recognition method in \cite{Liu23TVT}, which achieves satisfactory recognition performance with limited dataset.

Although abovementioned schemes can effectively reduce the complexity of precoding or reconstruct the spatial non-stationary channel with the aid of VR information, however, to the best of our knowledge, the impact of VR estimation error on the performance of XL-MIMO systems has not been discussed. Motivated by this, we consider a spatial non-stationary XL-MIMO channel, where the VRs of several single-antenna users are overlapped. We study the average achievable {energy efficiency (EE)} in the presence of VR detection error and then maximize {EE} by designing VR detection scheme and the length of uplink pilot. Our contributions are summarized as follows:
\begin{itemize}
\item In the uplink training phase, we propose a VR detection method by measuring the energy of received pilot signals, and derive two error probabilities for VR detection, i.e., the false detection and the missed detection probabilities, which are shown to be mainly affected by the detection threshold and the length of uplink pilot. Based on the detected VR, we further investigate the impact of VR detection error on the channel estimation performance, for both time division duplexing (TDD) and frequency division duplexing (FDD) modes.

\item According to the estimated imperfect channel state information (CSI) and adopting regularized zero-forcing (RZF) precoder, we derive the deterministic approximate expression for the ergodic average {EE} under TDD and FDD modes, respectively. Invoking these approximations, we formulate an optimization problem to maximize the average {EE} by designing VR detection threshold and uplink pilot length.

\item We propose an alternate optimization algorithm to jointly design the optimal detection threshold and pilot length for VR detection. To reduce complexity, we further propose heuristic solutions for the VR detection threshold with given pilot length, and their effectiveness is validated by numerical simulations.
\end{itemize}

The rest of this paper is organized as follows: In Section \ref{System Model}, we describe the system model. In Section \ref{es TDD}, we describe the procedure of VR detection and uplink channel estimation in TDD systems. Furthermore, the {EE} maximization problem is formulated. In Section \ref{alternate optimization SE}, we derive a deterministic approximation of the ergodic downlink average {EE} and propose an alternate optimization algorithm to design the detection threshold and the length of uplink pilot. In Section \ref{FDD_cond}, we further provide the deterministic approximation and the maximization of the average {EE} in FDD systems. Section \ref{simulations} reveals the impacts of VR estimation error on average {EE} and demonstrates the effectiveness of our proposed schemes by numerical results. Section \ref{conclusions} concludes our paper.

\emph{Notations}: The bold uppercase $\qX$ and bold lowercase $\qx$ denote matrices and vectors, respectively, and the superscripts in top right corner, i.g., $(\cdot)^{T},(\cdot)^{H},(\cdot)^{-1},(\cdot)^{\frac{1}{2}}$, stand for transpose, conjugate-transpose, inverse, and root of the matrix or vector, respectively. The sign $\bbE\{\cdot\}$ means the expectation result. $[\qX]_{i,j}$ indicates the element joined by the $i$-th row and $j$-th column of $\qX$. ${\sf tr}(\qX)$ denotes the trace of $\qX$ and $|\calS|$ denotes the ensemble of set $\calS$.

\section{System Model}\label{System Model}
We consider an XL-MIMO scenario where one BS is equipped with an $M$-antenna uniform linear array (ULA), serving multiple single-antenna users. As a result of the spatial non-stationarity, only part of the ULA can receive the signal from a user, which is described as VR. Usually, the VR varies for different users and we denote the visible antenna indices set of the $k$-th user by $\calA_{k} \subset \{1,2,\ldots,M\}$. The visible antennas in VR can be much fewer than $M$ in the whole array, which helps reduce the complexity of precoding significantly by grouping users properly \cite{Yang20TVT}. Since we assume that the BS has no prior knowledge on users' VR but acquires users' locations, we can group users according to their locations, due to the fact that adjacent users share the same scattering clusters and thus leads to the same VR at the BS \cite{Liu23TVT,Liu12WC}. Furthermore, for different non-overlapping VRs, there is no interference between different groups \cite{Ali19WC}. Thus, BS can communicate with different groups with their own VR. In this paper, we focus on one group where $K$ ($K\ll M$) users share the same VR to investigate the influences of VR knowledge on the ergodic sum-rate and average {EE}, i.e., we assume $\calA_{k}=\calA$ and $|\calA|=L$ for the $K$ users within one group, which is shown in Fig. \ref{System}. Note that the VR length of different user groups can be different.

\begin{figure}
\includegraphics[width=0.35\textwidth]{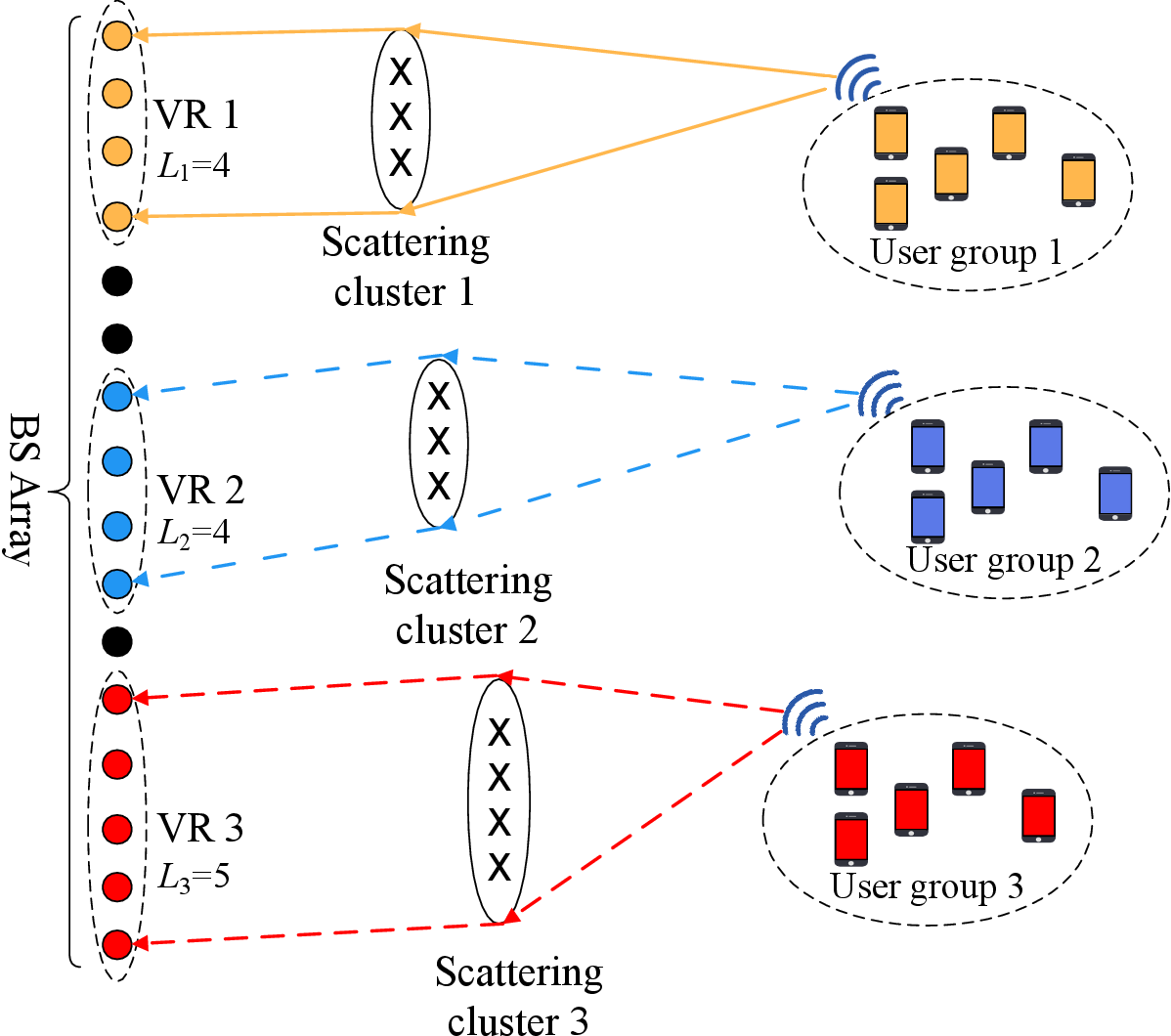}
\centering
\caption{{The spatially non-stationary XL-MIMO system.}} \label{System}
\end{figure}

The channel between the BS and the $k$-th user is modeled as\footnote{For the Rician model with line of sight (LoS) component, the VR is composed of two parts, i.e., the VR of LoS component and the VR of non-LoS component, which are different due to different propagation paths. The former can be precisely detected with the methods proposed in \cite{Tian23WCL,Han22JSTSP} since the LoS component is deterministic. Then the VR can be further detected from those non-LoS antennas with our proposed energy detector, and the probability analysis is similar to that in Section III. For simplicity, we consider only non-LoS component here to better reveal the impact of imperfect VR detection.}
\begin{align}
\qh_{k}^{\sf full}={\qR_{k}^{\sf full}}^{\frac{1}{2}}\qz_{k}^{\sf full},
\end{align}
where $\qz_{k}^{\sf full}\in\mathbb{C}^{M\times 1}$ denotes the fast fading channel vector with respect to $\mathcal{CN}(0,\frac{1}{M}\qI)$, the diagonal matrix $\qR_{k}^{\sf full}=\beta_k{\sf diag}\{\delta_{1,k},\delta_{2,k},\ldots,\delta_{M,k}\}$, $\delta_{m,k}$ is a binary indicator defined by
\begin{align}
\delta_{m,k}=&\left\{                     \begin{array}{ll}
                         0,m\notin\calA_k; \\
                         1,m\in\calA_k,
                     \end{array}
                        \right.\label{binary}
\end{align}
and $\beta_k$ is the large-scale fading known a priori to the users, which is defined as
\begin{align}
\beta_k=\beta_{0}d_k^{-\alpha},
\end{align}
where $\beta_{0}$ is the path loss at the reference distance, $d_k$ is the distance from the $k$-th user to the BS, and $\alpha$ is the path loss factor.

Many existing works assume that VR knowledge (i.e., $\calA_k$) is perfectly known in the BS. It is not clear how the system performance is impacted by imperfect VR knowledge. Therefore, we will reveal the impact of VR estimation error in the next section.

\section{Achievable Energy Efficiency of TDD Systems}\label{es TDD}

In this section, we first detect VR by uplink training before we perform channel estimation and analyze the achievable ergodic downlink {EE} by considering VR estimation error and linear precoding. Furthermore, we formulate the maximization problem for ergodic average {EE}, where detection threshold and uplink pilot length are joint optimized.

\subsection{Uplink VR Estimation and Probability Analysis}\label{VRest}
\begin{figure}[htb]
\includegraphics[width=0.45\textwidth]{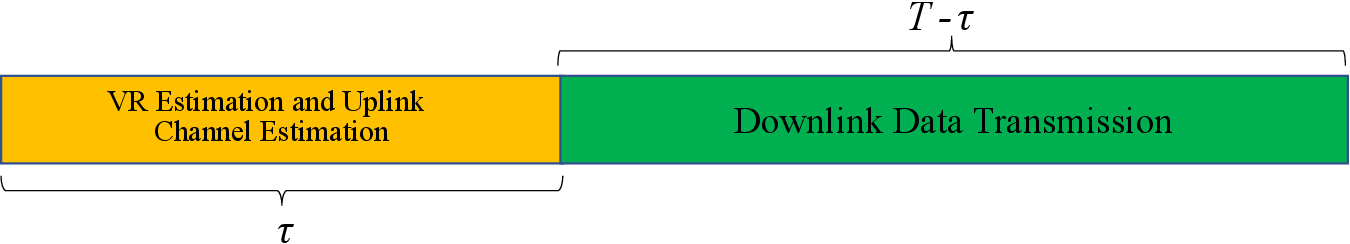}
\centering
\caption{The communication procedure of TDD systems.} \label{TDDpilot}
\end{figure}

In TDD systems, the channel coherence time is divided into two stages for uplink training and downlink data transmission as illustrated in Fig. \ref{TDDpilot}. {The length of available coherence time block is $T$ symbols}. During the uplink training stage, $K$ orthogonal pilot sequences are assigned to users\footnote{ We assume that the orthogonal pilot sequences are assigned per VR user group, and the same pilot sequences are reused among the different user groups since the VR of different groups are considered to be non-overlapping. For the scenario where $K>\tau$, we may assign those additional users with different time interval and reuse the pilot sequences in different intervals.} and satisfy $\qPhi^H \qPhi = \tau M\qI$, where $\qPhi = [\qphi_1,\qphi_2,\ldots,\qphi_K] \in \bbC^{\tau \times K}$ is the pilot matrix and $\tau$ is the pilot length. Thus, the received training signal $\qY_{\sf Tr} \in \bbC^{M\times \tau}$ in the BS is given by
\begin{equation}\label{eq:the received training signal}
 \qY_{\sf Tr} = \qH^{\sf full} \qP_{\sf Tr}^{\frac{1}{2}} \qPhi^H + \qN_{\sf Tr},
\end{equation}
where $\qH^{\sf full} = [\qh_1^{\sf full},\qh_2^{\sf full},\ldots,\qh_K^{\sf full}] \in \bbC^{M\times K}$ is the uplink channel, $\qP_{\sf Tr} = \diag\{p_{{\sf Tr},1},p_{{\sf Tr},2}, \ldots, p_{{\sf Tr},K}\}$ denotes the power of pilot signal from all users, and $\qN_{\sf Tr} \in \bbC^{M\times \tau}$ is the noise matrix whose entries are independent and identically distributed (i.i.d.) with respect to $\mathcal{CN}(0,\sigma_{\sf Tr}^2)$.

During the uplink training stage, the received signal correlated with $\qPhi$ at the BS can be given by
\begin{equation}\label{eq:the received training signal with the pilot sequences}
 \tqY_{\sf Tr} = \tau\sqrt{M} \qH^{\sf full} \qP_{\sf Tr}^{\frac{1}{2}} + \tqN_{\sf Tr},
\end{equation}
where $\tqY_{\sf Tr} = \frac{1}{\sqrt{M}}\qY_{\sf Tr} \qPhi \in \bbC^{M\times K}$ and $\tqN_{\sf Tr} = \frac{1}{\sqrt{M}}\qN_{\sf Tr}\qPhi \in \bbC^{M\times K}$.

From \eqref{eq:the received training signal with the pilot sequences}, the VR $\calA_k$ for different users can be detected respectively by judging the energy strength of observed signal $\tqY_{\sf Tr}$. We denote the ($m$,$k$)-th entry of matrix $\tqY_{\sf Tr}$ by $\ty_{m,k}^{\sf Tr}$ and it can be easily observed that $\ty_{m,k}^{\sf Tr}$ is also a complex Gaussian variable with a probability density function (PDF) given as
\begin{align}
f(\ty_{m,k}^{\sf Tr} | \delta_{m,k}) =& \frac{1}{\pi(\tau^2 p_{{\sf Tr},k} \beta_k \delta_{m,k} + \tau \sigma_{\sf Tr}^2)} \nonumber \\
 &\times \exp \left\{- \frac{|\ty_{m,k}^{\sf Tr}|^2}{\tau^2 p_{{\sf Tr},k} \beta_k \delta_{m,k} + \tau \sigma_{\sf Tr}^2}\right\}. \label{eq:ty PDF}
\end{align}
That is
\begin{equation}\label{eq:ty distributed}
 \ty_{m,k}^{\sf Tr} \sim \left\{\begin{aligned}
& \mathcal{CN}(0,\tau^2 p_{{\sf Tr},k} \beta_k + \tau \sigma_{\sf Tr}^2), & &  \delta_{m,k} = 1; \\
& \mathcal{CN}(0,\tau \sigma_{\sf Tr}^2), & & \delta_{m,k} = 0.
\end{aligned}
\right.
\end{equation}

The VR indicator $\delta_{m,k}$ can not be detected directly from $\ty_{m,k}^{\sf Tr}$ since the deterministic fast fading channel knowledge is difficult to obtain in XL-MIMO. Therefore, we adopt the following energy detector described by
\begin{equation}\label{eq:energy detector}
 \hat{\delta}_{m,k} =\left\{\begin{aligned}
& 1, & & |\ty_{m,k}^{\sf Tr}|^2 > \zeta_0; \\
& 0, & & |\ty_{m,k}^{\sf Tr}|^2 \leq \zeta_0,
\end{aligned}
\right.
\end{equation}
where $\zeta_0$ denotes the detection threshold.

The energy detector \eqref{eq:energy detector} will incur two error events for VR detection: 1) false detection, i.e., $\hat{\delta}_{m,k} = 1$ when $\delta_{m,k} = 0$; 2) missed detection, i.e., $\hat{\delta}_{m,k} = 0$ when $\delta_{m,k} = 1$, where the probabilities of the two conditions can be calculated as follows, respectively,
\begin{align}
 P_{10,k} = &{\sf Pr}(\hat{\delta}_{m,k}=1|\delta_{m,k}=0)= \exp\left\{-\frac{\zeta_0}{\tau \sigma_{\sf Tr}^2}\right\},  \label{eq:Pr10}\\
 P_{01,k} = & {\sf Pr}(\hat{\delta}_{m,k} = 0|\delta_{m,k}=1) \nonumber \\
          = &1- \exp\left\{-\frac{\zeta_0}{\tau^2 p_{{\sf Tr},k} \beta_k + \tau \sigma_{\sf Tr}^2}\right\}.  \label{eq:Pr01}
\end{align}

For analytic tractability, we let $p_{{\sf Tr},1} \beta_1 = \ldots = p_{{\sf Tr},K} \beta_K$ in \eqref{eq:Pr01} through pilot power compensation for different users by efficient calibrations \cite{Bjo15TWC}. Let $p_{{\sf Tr},1} \beta_1=g$ and the probabilities of two error events can be rewritten as
\begin{align}
 P_{10} = & {\sf Pr}(\hat{\delta}_{m,k} = 1|\delta_{m,k}=0) = \exp\left\{-\frac{\zeta_0}{\tau\sigma_{\sf Tr}^2}\right\},  \label{eq:Pr10-2}  \\
 P_{01} = & {\sf Pr}(\hat{\delta}_{m,k} = 0|\delta_{m,k}=1) \nonumber \\
        = & 1- \exp\left\{-\frac{\zeta_0}{\tau^2 g + \tau \sigma_{\sf Tr}^2}\right\},  \label{eq:Pr01-2}
\end{align}
respectively.

Furthermore, the detected VR set of $\calA$ is denoted as $\hat{\calA}$ and it satisfies $|\hat{\calA}| = I+J$, where $I\ (0 \leq I \leq L)$ is assumed to be the number of correctly estimated antennas (i.e., $|\hat{\calA} \cap \calA| = I$) and $J\ (0\leq J \leq M - L)$ is assumed to be the number of wrongly estimated antennas.
Thus, the probability of $\hat{\calA}$ with arbitrary $\{I,J\}$ can be calculated as
\begin{align}
 P_{I,J} = & {\sf Pr}(\hat{\calA}|\calA) \nonumber \\
         = &  \binom{L}{I}\binom{M-L}{J}P_{01}^{L-I} (1-P_{01})^{I} P_{10}^{J} (1-P_{10})^{M-L-J}, \label{eq:PrIJ}
\end{align}
where $I=0,1,2,\ldots,L$ and $J=0,1,2,\ldots,M-L$. Notice that $\hat{\calA} = \calA$ if $I=L$ and $J=0$, which suggests that the VR estimation is exactly correct.

\subsection{Uplink Channel Estimation Based on $\hat{\calA}$}
With detected VR, we can perform uplink channel estimation only on $\hat{\calA}$ rather than the whole array, which helps significantly reduce the computation complexity. We denote the estimated channel vector between the $k$-th user and the BS on $\hat{\calA}$ as $\hat{\qh}_{\hat{\calA},k}\in\mathbb{C}^{(I+J)\times 1}$, and similarly the real channel vector is denoted as $\qh_{\hat{\calA},k}\in\mathbb{C}^{(I+J)\times 1}=[\qh^{\sf full}_k]_{\hat{\calA}}$. Since there are $J$ wrongly estimated antennas in set $\hat{\calA}$, the real channel $\qh_{\hat{\calA},k}\in\mathbb{C}^{(I+J)\times 1}$ has $J$ zero elements.

If the least-square (LS) method is employed for uplink channel estimation, we have the estimated channel by
\begin{align}
\hat{\qH}_{\hat{\calA},\sf LS}=&\qH_{\hat{\calA}}+\frac{1}{M\tau}[\mathbf{N}_{\sf Tr}]_{\hat{\calA}}\qPhi\qP_{\sf Tr}^{-\frac{1}{2}} \nonumber \\
       =&[\hat{\qh}_{\hat{\calA},1},\ldots,\hat{\qh}_{\hat{\calA},K}],
\end{align}
where $\hat{\qh}_{\hat{\calA},k}=\qh_{\hat{\calA},k}+\qe_{\hat{\calA},k}$, and $\qe_{\hat{\calA},k}=\frac{1}{\sqrt{p_{{\sf Tr},k}}M\tau}[\mathbf{N}_{\sf Tr}]_{\hat{\calA}}\qphi_{k}$ is the equivalent estimation error with respect to $\mathcal{CN}(\mathbf{0},\frac{\sigma^2_{\sf Tr}}{p_{{\sf Tr},k}M\tau}\qI_{I+J})$.

\subsection{Achievable Downlink {Energy Efficiency}}
In downlink transmission stage, the BS transmits the user signal with the $(I+J)$ antennas in detected VR $\hat{\calA}$ from the $(T-\tau+1)$-th to $T$-th time slot. The received signal at the $k$-th user can be expressed as
\begin{align}
y_{k}=\qh^{H}_{\hat{\calA},k}\qx+n_{k},
\label{signal}
\end{align}
where $n_{k}$ is the additive white Gaussian noise with respect to $\mathcal{CN}(0,\sigma^2_{k})$ and the transmitted signal $\qx\in\mathbb{C}^{(I+J)\times 1}$ can be given by
\begin{align}
\qx&=\qW\qP_{\sf DL}^{\frac{1}{2}}\qs =\sum_{k=1}^{K}\sqrt{p_{{\sf DL},k}}\qw_{k}s_{k},\label{xsignal}
\end{align}
where $\qP_{\sf DL}={\sf diag} \{p_{{\sf DL},1}, \ldots, p_{{\sf DL},K}\}$, $\qW=[\qw_{1},\ldots,\qw_{K}]$, $\qs^{T}=\{s_1,\ldots,s_K\}$, $p_{{\sf DL},k}, s_{k}$, and $\qw_{k}\in\mathbb{C}^{(I+J)\times 1}$ are the power scaling factor, i.i.d. user symbol with zero mean and unit variance, and the precoding vector for the $k$-th user, respectively.
{The transmitted signal follows the power constraint given below.}
\begin{align}\label{consP}
\mathbb{E}\{\qx^{H}\qx\}&={\sf tr}(\qW^{H}\qW\qP_{\sf DL}) =MP_{T},
\end{align}
where $P_T$ is the total transmit power in the BS and the multiplication of $M$ is due to the normalization for channel realization $\qh$.

In this paper, we adopt the RZF precoding \cite{Joham02ISSSTA} for downlink transmission. For ease of notation, we simply refer to $\hat{\qH}_{\hat{\calA},{\sf LS}}$ as $\hat{\qH}_{\hat{\calA}}=[\hat{\qh}_{\hat{\calA},1},\ldots,\hat{\qh}_{\hat{\calA},K}]$ hereafter in this paper. Then the RZF precoding matrix can be written as
\begin{align}\label{RZF}
\qW = \alpha(\hat{\qH}_{\hat{\calA}}\hat{\qH}_{\hat{\calA}}^{H} + \xi\qI_{I+J})^{-1}\hat{\qH}_{\hat{\calA}} = \alpha\qA^{-1}\hat{\qH}_{\hat{\calA}},
\end{align}
where $\alpha>0$ is the power normalization scalar and $\xi>0$ is the regularization scalar. Note that when $\xi\rightarrow0$ and $\xi\rightarrow\infty$, the RZF precoding degrades into ZF precoding and CB precoding, respectively. To satisfy the transmit power constraint \eqref{consP}, $\alpha>0$ is set as
\begin{align}\label{alpha}
\alpha=\sqrt{\frac{MP_{T}}{\mathbb{E}\{{\sf tr}(\qP_{\sf DL}\hat{\qH}_{\hat{\calA}}^{H}\qA^{-2}\hat{\qH}_{\hat{\calA}})\}}}.
\end{align}

With the structure of \eqref{xsignal}, the received signal can be rewritten as
\begin{align}
y_{k}=\sqrt{p_{{\sf DL},k}}\qh_{\hat{\calA},k}^{H}\qw_ks_k+\sum_{j\neq k}^{K}\sqrt{p_{{\sf DL},j}}\qh_{\hat{\calA},k}^{H}\qw_js_j+n_k.
\end{align}
Then, the signal to interference plus noise ratio (SINR) in the $k$-th user can be given by
\begin{align}\label{ergSINR}
\gamma_{\hat{\calA},k}=\frac{p_{{\sf DL},k}|\qh_{\hat{\calA},k}^{H}\qw_{k}|^2}{\sum_{j\neq k}^{K}p_{{\sf DL},j}|\qh_{\hat{\calA},k}^{H}\qw_{j}|^2+\sigma^{2}_{k}}.
\end{align}
Therefore, the ergodic downlink sum-rate can be given by
\begin{align}\label{ergR}
R^{\sf sum}_{I,J}=\sum_{k=1}^{K}\bbE\{\log(1+\gamma_{\hat{\calA},k})\}.
\end{align}
Notice that the VR estimation result $\hat{\calA}$ can lead to significant difference on $R^{\sf sum}_{I,J}$, where $\{I,J\}$ varies with probability $P_{I,J}$. Note that $P_{I,J}$ is the same for different users since they share the same VR and the VR detection is performed for all users within this group. {Thus, the ergodic average EE can be expressed as
\begin{align}\label{ergSE}
{\sf EE}_{\sf TDD}=\frac{W(T-\tau)}{T}\sum_{I=0}^{L}\sum_{J=0}^{M-L}P_{I,J}\frac{R_{I,J}^{\sf sum}}{P_{\sf cir}}.
\end{align}
where $P_{\sf cir}$ is the consumed power for downlink transmission and $W$ is the bandwidth. $P_{\sf cir}$ can be expressed as \cite{Saba18TWC}
\begin{align}
P_{\sf cir}=\varsigma^{-1}MP_T+2P_{\sf syn}+(I+J)P_{\sf CT}+KP_{\sf CR},
\end{align}
where $\varsigma$ is the efficiency of power amplifier, $P_{\sf syn}$ is the power of frequency synthesizer, $P_{\sf CT}$ and $P_{\sf CR}$ are the circuit power consumed by each RF chain at the transmitter and receiver, respectively.
}

{\emph{Remark 1:} Apart from the fully overlapping VR condition we consider here, it can be very interesting to explore the imperfect VR detection under the partially overlapping VR condition. We may start from a specific and easy case, where $M=16, K=4$ and different $L$ for different user (denoted by $L_{1},L_{2},L_{3},L_{4}$, respectively). The specific true VR distribution is given in Fig. \ref{sp_case}.
\begin{figure}[htb]
\includegraphics[width=0.4\textwidth]{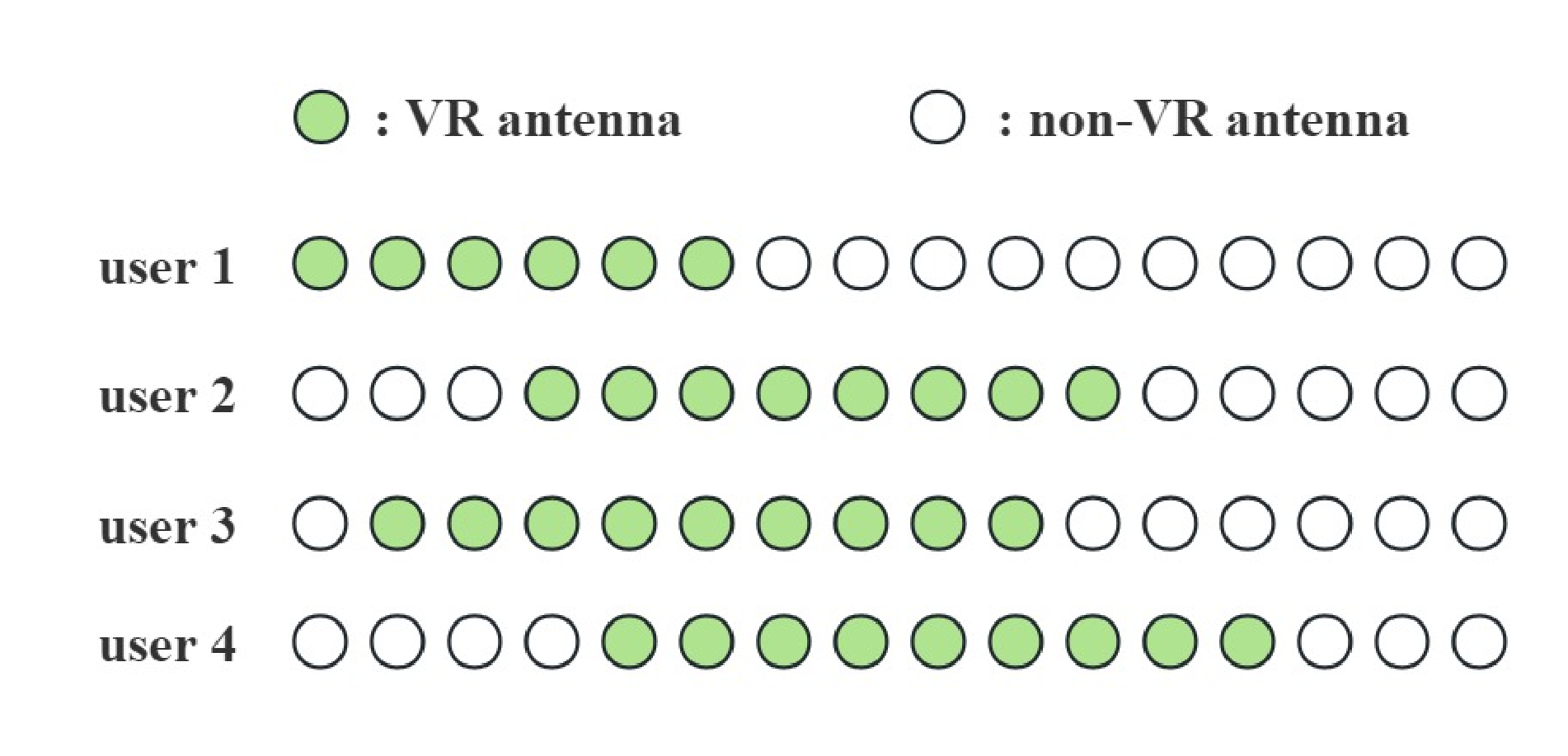}
\centering
\caption{A specific case of partially overlapping VR with $M=16$ and $K=4$.} \label{sp_case}
\end{figure}

Note that the VRs of different users must be detected independently since their VRs are not the same. Therefore, it is necessary to enumerate and calculate the probability for all possible results since average EE is not only concerned with the number of wrongly/correctly detected antennas. By assuming power compensation and i.i.d. channel for different users, the VR detection on each antenna for a single user can be either correct or wrong, which can be represented by 1 or 0, respectively. Thus we can use $P_{0000,m}$ to represent the probability of the event that ``the VR detections on the $m$-th antenna are totally wrong for all 4 users", and $P_{0001,m}$ to represent the probability of the event that ``the VR detections on the $m$-th antenna are wrong for user 1,2,3 but correct for user 4", and so on. Assuming ``0000" and ``0001" being binary codes, we redefine $P_{0000,m}, P_{0001,m}, P_{0010,m}, \ldots, P_{1111,m}$ by $P_{0,m}, P_{1,m}, P_{2,m}, \ldots, P_{15,m}$, respectively. Obviously, there are $2^K=2^4=16$ possible results of VR detection for each antenna. Specifically, considering the first and the second antenna from the left, we have following probabilities
\begin{align}
P_{b_1}=&P_{01}^{1-b_1^{(1)}}(1-P_{01})^{b_1^{(1)}}P_{10}^{3-b_1^{(2)}-b_1^{(3)}-b_1^{(4)}}\nonumber\\
&\times(1-P_{10})^{b_1^{(2)}+b_1^{(3)}+b_1^{(4)}},\label{P_b1}\\
P_{b_2}=&P_{01}^{2-b_2^{(1)}-b_2^{(3)}}(1-P_{01})^{b_2^{(1)}+b_2^{(3)}}\nonumber\\
&\times P_{10}^{2-b_2^{(2)}-b_2^{(4)}}(1-P_{10})^{b_2^{(2)}+b_2^{(4)}},\label{P_b2}
\end{align}
where $b_1$ and $b_2$ are integers between 0 and 15, indicating one VR detection result for the first and second antenna, respectively, and $b^{(i)}$ denotes the $i$-th bit of the corresponding binary code for $b$ (counting from the left). From \eqref{P_b1} and \eqref{P_b2}, we can see that the probabilities of VR detection for different antennas are concerned with the true VR distribution, and the expression becomes more and more complicate as $K$ grows. Based on \eqref{P_b1} and \eqref{P_b2}, we have the probability of any possible VR detection result by $\prod_{m=1}^{M}P_{b_m}$, and the average EE under the partially overlapping VR condition can be expressed as
\begin{align}
{\sf EE}_{\sf partVR}=\frac{W(T-\tau)}{T}\sum_{b_1=0}^{2^K-1}P_{b_1}\dots\sum_{b_M=0}^{2^K-1}P_{b_M}\frac{R_{b_1,\ldots,b_M}^{\sf sum}}{P_{\sf cir}},
\end{align}
where $R_{b_1,\ldots,b_M}^{\sf sum}$ is the sum-rate under VR detection $\{b_1,\ldots,b_M\}$. Obviously, there are $2^{KM}=2^{64}$ possible VR sets in total. Unfortunately, it is difficult to find a common feature, e.g., $\{I,J\}$ and $P_{I,J}$ we use in \eqref{eq:PrIJ}, that helps divide massive possible $\hat{\calA}$ into a few categories. Furthermore, to obtain better performance and lower complexity, it is necessary to elaborate the precoder according to the detected VR. Based on above analysis, the average performance with imperfect VR detection under the partially overlapping VR scenario remains an open problem.}


\subsection{Problem Formulation}
Our target is to maximize the ergodic average {EE} by optimizing the detection threshold $\zeta_0$ and uplink pilot length $\tau$. Therefore, the problem can be formulated as
\begin{align}
(\mathcal{P}1)\ \max_{\zeta_0>0,\tau=K,\ldots,T} {\sf EE}_{\sf TDD}.
\end{align}
Problem $(\mathcal{P}1)$ is difficult to solve due to the ergodic objective function, which requires massive channel realizations to operate Monte-Carlo averaging. To avoid the high complexity, we can utilize the deterministic approximations for \eqref{ergSE} to measure the objective function of problem $(\mathcal{P}1)$, thus $\{\zeta_0, \tau\}$ can be determined with much lower complexity.

\section{Optimization of Detection Threshold and Uplink Pilot Length }\label{alternate optimization SE}

In this section, a deterministic approximate expression for ergodic {EE} is derived, which allows measuring performances without performing massive Monte-Carlo trials. Based on the deterministic approximate expression for \eqref{ergSE}, we turn to solve an equivalent problem without performing Monte-Carlo averaging, and propose an alternate optimization algorithm to design $\zeta_0$ and $\tau$.

\subsection{Deterministic Approximate Results for \eqref{ergSE} }

The ergodic SINR makes it difficult to analyze achievable {EE} in finite-dimensions system. Therefore, we present a deterministic approximation for { EE} by considering the large-system assumption, where $M$ and $K$ both grow to infinity at fixed ratio $M/K$ that $0<\lim\inf M/K\leq\lim\sup M/K<\infty$ \cite{Jun13TWC}. We use $\mathcal{N}\rightarrow\infty$ to represent the above large-system assumption for brevity. The approximation is given in Theorem 1 as follows
\begin{theorem}\label{Theorem1}
As $\mathcal{N}\rightarrow\infty$, we have $\gamma_{\hat{\calA},k}\stackrel{\sf a.s.}{\longrightarrow}\bar{\gamma}_{\hat{\calA},k}$, where $\bar{\gamma}_{\hat{\calA},k}$ is given by
\begin{align}
\bar{\gamma}_{\hat{\calA},k}=\frac{\frac{1}{M} p_{{\sf DL},k}\mu_{k}^2}{\lambda_k(1+a_k\bar{\mu})^2+\frac{1}{M}a_k\bar{\lambda}_k\mu_{k}^2+\frac{\bar{\lambda}\sigma_{k}^2}{P_T}(1+a_k\bar{\mu}+\frac{1}{M}\mu_k)^2},
\end{align}
where $a_{k}=\frac{\sigma^2_{\sf Tr}}{M\tau p_{{\sf Tr},k}}$, and $\mu_k$, $\lambda_k$, $\bar{\mu}$, $\bar{\lambda}_{k}$, $\bar{\lambda}$ are respectively given by
\begin{align}
\mu_{k}&=\frac{I\beta_{k}}{\xi+\frac{\eta}{M}(1+\frac{\sigma_{\sf Tr}^2}{g\tau})},\label{mu_k}\\
\lambda_{k}&=\sum_{i\neq k}^{K}p_{{\sf DL},i}\frac{e_{i,k}'}{(1+e_i)^2},\label{lambda_k}\\
\bar{\mu}&=\frac{I}{\xi+\frac{\eta}{M}(1+\frac{\sigma_{\sf Tr}^2}{g\tau})}+\frac{J}{\xi+\frac{\sigma^2_{\sf Tr}\eta}{M\tau g}},\label{bar_mu}\\
\bar{\lambda}_{k}&=\sum_{i\neq k}^{K}p_{{\sf DL},i}\frac{e_{i}'}{(1+e_i)^2},\label{bar_lambda_k}
\end{align}
\begin{align}
\bar{\lambda}=\sum_{i=1}^{K}p_{{\sf DL},i}\frac{e_{i}'}{(1+e_i)^2}\label{bar_lambda}.
\end{align}
The parameters $\eta=\sum_{i=1}^{K}\frac{\beta_i}{1+e_i}$, $e_{i,k}'=\Big[(\qI_{K}-\qJ)^{-1}\qv_{k}\Big]_i$, $e_{i}'=\Big[(\qI_{K}-\qJ)^{-1}\qv'\Big]_i$, and
\begin{align}
e_i=\beta_i\Big(\frac{I(1+\frac{\sigma_{\sf Tr}^2}{g\tau})}{M\xi+\eta(1+\frac{\sigma_{\sf Tr}^2}{g\tau})}+\frac{J}{M\xi\frac{g\tau}{\sigma^2_{\sf Tr}}+\eta}\Big),\label{e}
\end{align}
where
\begin{align}
[\qJ]_{i,j}&=\frac{\beta_{i}\beta_{j}}{(1+e_j)^2}\Big(\frac{I(1+\frac{\sigma_{\sf Tr}^2}{g\tau})^2}{(M\xi+\eta(1+\frac{\sigma_{\sf Tr}^2}{g\tau}))^2}+\frac{J}{(M\xi\frac{g\tau}{\sigma^2_{\sf Tr}}+\eta)^2}\Big),\label{qJ}\\
[\qv_{k}]_i&=(1+\frac{\sigma_{\sf Tr}^2}{g\tau})\frac{I\beta_i\beta_k}{M(\frac{\eta}{M}+\frac{\eta\sigma_{\sf Tr}^2}{M g\tau}+\xi)^2},\label{qv}\\
[\qv']_{i}&=\frac{\beta_{i}}{M}\Big(\frac{I(1+\frac{\sigma_{\sf Tr}^2}{g\tau})}{(\frac{\eta}{M}+\frac{\sigma_{\sf Tr}^2\eta}{Mg\tau}+\xi)^2}+\frac{J\frac{\sigma_{\sf Tr}^2}{g\tau}}{(\xi+\frac{\sigma_{\sf Tr}^2\eta}{Mg\tau})^2}\Big).\label{qv'}
\end{align}
\end{theorem}
\begin{proof}
See Appendix \ref{The1proof}.
\end{proof}
Based on Theorem \ref{Theorem1}, we have several remarks as follow.

\emph{Remark 2:} As $\xi$ tends to be infinite and $I>0$, RZF precoder degrades into CB precoder, i.e., $\qW=\alpha\hat{\qH}_{\hat{\calA}}$. Thus, the deterministic result reduces to \eqref{Remark 2}, as shown at the top of next page,
\begin{figure*}[!t]
\begin{equation}\label{Remark 2}
\bar{\gamma}_{\hat{A},k} =\frac{p_{{\sf DL},k}I\beta_k^2}{\beta_k(1+\frac{\sigma_{\sf Tr}^2}{g\tau})\sum_{j\neq k}^{K}p_{{\sf DL},j}\beta_j+\frac{\sigma^2_{k}}{P_T}\sum_{j=1}^{K}p_{{\sf DL},j}\beta_j(1+\frac{\sigma^2_{\sf Tr}}{g\tau}(1+\frac{J}{I}))},
\end{equation}
\hrulefill
\end{figure*}
which is identical with the result obtained in \cite{Zhang22WCSP}. {We observe that the impact of $J$ is limited with high uplink SNR (i.e., $\frac{g\tau}{\sigma^2}$) or high downlink SNR (i.e., $\frac{P_T}{\sigma^2_k}$), whereas $I$ has direct and significant influence on SINR. Furthermore, we notice that the impact of $J$ is concerned with $\frac{J}{I}$, i.e., the impact of $J$ becomes more significant when $I$ is small. Based on these analysis, we can draw conclusion that it is more important to avoid missed detection, and the impact of false detection on SINR can be effectively reduced if most VR antennas are detected correctly. }

\emph{Remark 3:} As $\xi$ and $g\tau$ both grow to infinity, and let $I=L,J=0$, the deterministic result reduces to
\begin{align}
\bar{\gamma}_{\calA,k}=\frac{p_{{\sf DL},k}L\beta_k^2}{\beta_k\sum_{j\neq k}^{K}p_{{\sf DL},j}\beta_j+\frac{\sigma^2_{k}}{P_T}\sum_{j=1}^{K}p_{{\sf DL},j}\beta_j},
\end{align}
which also agrees with the result of \cite[eq.(11)]{Ali19WC}.

\emph{Remark 4:} Notice that VR information $\hat{\calA}$ changes slowly in a long period if scatters are static, thus VR estimation can be performed only once during several coherence blocks. Note that simply accumulating pilot signals from different coherence blocks for VR detection will not help improve the accuracy of estimation since noise is also accumulated.
For a coherence block without VR estimation, the average {EE} can expressed as
\begin{align}
{\bar{{\sf EE}}_{\sf TDD} = \frac{W(T-\tau)}{T}\sum_{I=0}^{L}\sum_{J=0}^{M-L}P_{I,J}^{(t)}\frac{\bar{R}^{\sf sum}_{I,J}}{P_{\sf cir}},}
\end{align}
where $P_{I,J}^{(t)}$ is fixed here, generated from the $t$-th coherence block that performs VR estimation and $\bar{R}_{I,J}^{\sf sum}=\sum_{k=1}^{K}\log(1+\bar{\gamma}_{\hat{\calA},k})$. Herein, the uplink pilot length $\tau$ only impacts the channel estimation error $\qe_{\hat{\calA},k}$.

\subsection{Alternate Optimization}\label{deterministic_appr}
Invoking the approximated expression in Theorem \ref{Theorem1}, we reformulate problem $(\calP1)$ as
\begin{align}
(\mathcal{P}2)\ \max_{\zeta_0>0,\tau=K,K+1,\ldots,T} \bar{{\sf EE}}_{\sf TDD}.
\end{align}
With the deterministic objective function, we can optimize $\{\zeta_0,\tau\}$ without calculating massive Monte-Carlo trials.

Problem $(\calP2)$ does not admit a closed-form solution for $\{\zeta_0,\tau\}$. Thus, we adopt the alternate optimization method on the parameters since they are separate on their constraints. By fixing one of $\{\zeta_0,\tau\}$, we can break problem $(\calP2)$ into two one-dimension subproblems.

\subsubsection{Optimization on $\zeta_0$}\label{zeta_opt}
Given $\tau$, we have following subproblem by dropping irrelevant terms
\begin{align}\label{subpzeta}
\zeta_0^{\sf opt}=\arg\max_{\zeta_0>0}\sum_{I=0}^{L}\sum_{J=0}^{M-L}P_{I,J}(\zeta_0){\frac{\bar{R}_{I,J}^{\sf sum}}{P_{\sf cir}}}.
\end{align}
The above problem is one-dimension, and it can be solved with common convex optimization tools. {However, the expression of the objective function is complex and $\zeta_0$ is couple with many variables. To further reduce the complexity of determining $\zeta_0$, we provide following heuristic solutions. Compared with solving \eqref{subpzeta}, the two solutions depend on fewer variables.}

a)~ Minimizing $P_{01}+P_{10}$:
  Typically, minimizing the sum probability of incorrect detections is utilized to choose the detection threshold, and the corresponding threshold can be given as
\begin{align}
\zeta_0=\frac{\sigma^2_{\sf Tr}(\tau g+\sigma^2_{\sf Tr})}{g}\ln\frac{\sigma^2_{\sf Tr}+\tau g}{\sigma^2_{\sf Tr}}.
\end{align}

b)~ Letting $P_{01}=P_{10}$:
  By setting $P_{01}=P_{10}$, we can obtain a detection threshold which minimizes $P_{01}$ and $P_{10}$ simultaneously. Thus, $\zeta_0$ can be obtained by {solving following equation}.
\begin{align}\label{P01=P10}
\zeta_0=\frac{\sigma_{\sf Tr}^2(\tau g+\sigma_{\sf Tr}^2)}{g}\ln(\exp(\frac{\zeta_0}{\tau\sigma^2_{\sf Tr}})-1).
\end{align}
{The numerical solution of \eqref{P01=P10} can be determined by alternately and iteratively updating the right and left side of the equation respectively until $\zeta_0$ achieves convergence.}

\subsubsection{Optimization on $\tau$}
Given $\zeta_0$, we can optimize $\tau$ by solving following subproblem.
\begin{align}\label{subptau}{\tau^{\sf opt}=\arg\max_{\tau=K,K+1,...,T}\frac{W(T-\tau)}{T}\sum_{I=0}^{L}\sum_{J=0}^{M-L}P_{I,J}(\tau)\frac{\bar{R}_{I,J}^{\sf sum}(\tau)}{P_{\sf cir}}.}
\end{align}
The optimal length of uplink pilot $\tau$ which maximizes the approximate {EE} can be identified by searching {$\tau=K,K+1,\ldots,T$. Specifically, we will calculate the numerical value of $\bar{\sf EE}_{\sf TDD}(\tau)$ for $\tau=K,K+1,\ldots,T$ and find $\tau^{\sf opt}$. To calculate $\bar{R}_{I,J}^{\sf sum}$ for one time, the complexity can be roughly measured by $\mathcal{O}(K^3)$ from the matrix inversion when calculating $e'_{i,k}$ and $e'_i$. Since $\bar{\sf EE}_{\sf TDD}(\tau)$ requires to calculate $\bar{R}_{I,J}^{\sf sum}$ for $L(M-L+1)$ times, the total complexity can be measured by $\mathcal{O}((T-K)L(M-L+1)K^3)$, which is acceptable because the scheme is based on statistical CSI.}

The procedure of proposed alternate optimization algorithm is summarized as {\bfseries Algorithm 1}.
\begin{algorithm}[htb]
\caption{Alternate Optimization Algorithm for Problem $(\calP2)$}\label{AOMethod}
\begin{algorithmic}[1]
\REQUIRE $\{\beta_k\}_{k=1,2,\ldots,K},P_T,\sigma^2_{\sf Tr},g,\{\sigma^2_{k}\}_{k=1,2,\ldots,K},\xi$.
\STATE \textbf{Initialization:} $i=0,\tau^{(0)}=K$;
\REPEAT
\STATE $i=i+1$;
\STATE Update $\zeta_0^{(i)}$ by solving \eqref{subpzeta};
\STATE Update $\tau^{(i)}$ by solving \eqref{subptau};
\UNTIL $\tau^{(i-1)}==\tau^{(i)}$;
\RETURN $\zeta_0^{(i)},\tau^{(i)}$;
\end{algorithmic}
\end{algorithm}

\section{Achievable {Energy }Efficiency of FDD Systems}\label{FDD_cond}
In FDD systems, the reduction of BS antennas is also meaningful since it helps save the cost of pilot training by employing only VR antennas. The complete communication procedure in an FDD system contains three stages, i.e., 1) VR estimation, 2) Downlink channel estimation, and 3) Downlink data transmission, as shown in Fig. \ref{FDDpilot}. Different from TDD mode, the channels are not reciprocal between uplink and downlink in FDD systems, where the BS obtains the downlink CSI through the feedback channel from users. {In this paper, we consider the downlink CSI is fed back to BS with quantization error through a finite feedback channel \cite{Wagner12TIT}}.

\begin{figure}
\includegraphics[width=0.45\textwidth]{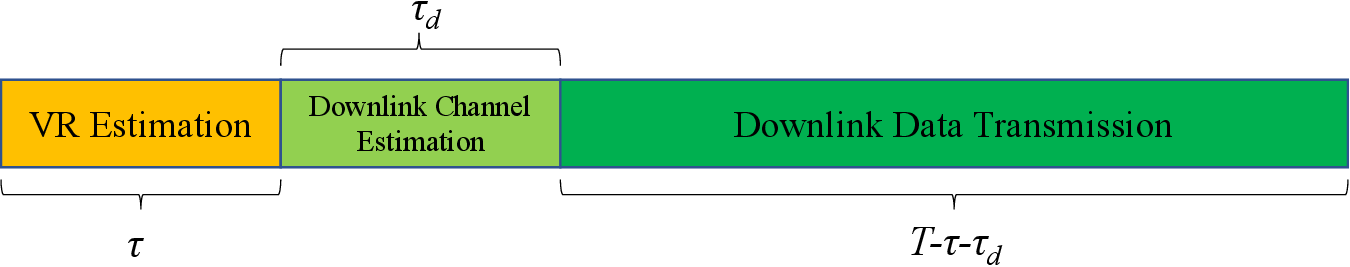}
\centering
\caption{The communication procedure of FDD systems.} \label{FDDpilot}
\end{figure}

The uplink VR estimation is similar to the VR estimation at TDD mode in Section \ref{VRest} and it leads to the same probabilities of two error detections. Thus, we consider the downlink channel estimation and data transmission.

\subsection{Downlink Channel Estimation in FDD Systems}
The BS sends unitary pilot sequences to all users at the detected VR $\hat{\calA}$. The received signal in the $k$-th user $\qy_{\hat{\calA},k}^H=\sqrt{p_{\sf dp}}\qh_{\hat{\calA},k}^{H}\qX+\qn^{H}_{\hat{\calA},k}$,
where $p_{\sf dp}$ is the power of downlink pilot and the pilot $\qX\in\bbC^{(I+J)\times\tau_d}$ satisfies $\qX\qX^{H}=\tau_dM\qI_{I+J}$.
$\tau_d\geq I+J$ is the length of downlink pilot. By adopting LS estimation, the estimated downlink channel can be expressed as $\hat{\qh}_{\hat{\calA},k}^{H}=\qh_{\hat{\calA},k}^{H}+\qe_{\hat{\calA},k}^{H}$,
where the estimation error $\qe_{\hat{\calA},k}\in\bbC^{(I+J)\times1}$ is with respect to $\mathcal{CN}(\mathbf{0},\frac{\sigma_{k}^2}{p_{\sf dp}M\tau_d}\qI_{I+J})$.

{The BS can obtain the downlink channel through a feedback channel, which incurs the quantization error. Thus, the channel information received at the BS can be expressed as \cite{Wagner12TIT,Jun13TWC}
\begin{align}
\tilde{\mathbf{h}}_{\hat{\mathcal{A}},k}=\sqrt{1-q^2}\hat{\mathbf{h}}_{\hat{\mathcal{A}},k}^{H}+q\mathbf{e}_{q},
\end{align}
where $0\leq q\leq1$ measures the quantization accuracy of feedback channel, and $\mathbf{e}_{q}\in\mathbb{C}^{(I+J)\times 1}$ is the random quantization error with respect to $\mathcal{CN}(0,\frac{\beta_k}{M}\mathbf{I})$. $q$ satisfies $q^2\leq2^{-\frac{B}{I+J-1}}$ \cite{Wagner12TIT}, where $B$ is the number of quantization bit for each user. In simulations, we set $q^2=2^{-\frac{B}{I+J-1}}$ for simplicity \cite{Wagner12TIT,Jun13TWC}.
}
\subsection{Downlink Data Transmission and Average {EE}}
With estimated channel {$\tilde{\qh}_{\hat{\calA},k}$}, we can adopt RZF precoding for downlink transmission and the ergodic average {EE} can be given as
\begin{align}
{{\sf EE_{FDD}}=W\sum_{I=0}^{L}\sum_{J=0}^{M-L}\frac{T-\tau-\tau_d}{TP_{\sf cir}}P_{I,J}R_{I,J}^{\sf FDD},}
\end{align}
where the ergodic sum-rate $R^{\sf FDD}_{I,J}$ in FDD systems can be expressed as
\begin{align}
R_{I,J}^{\sf FDD}=\sum_{k=1}^{K}\bbE\{\log(1+\gamma^{\sf FDD}_{\hat{\calA},k})\},
\end{align}
and the SINR $\gamma^{\sf FDD}_{\hat{\calA},k}$ of FDD systems has the same structure as \eqref{ergSINR} in TDD systems.
Similarly, to design the detection threshold $\zeta_0$ and uplink pilot length $\tau$, we give the deterministic approximation for the downlink SINR in FDD systems as follow.
{\begin{theorem}\label{Theorem1FDD}
As $\calN\rightarrow\infty$, we have
\begin{align}
\gamma^{\sf FDD}_{\hat{\calA},k}-\bar{\gamma}^{\sf FDD}_{\hat{\calA},k}\stackrel{\sf a.s.}{\longrightarrow}0,
\end{align}
where $\bar{\gamma}^{\sf FDD}_{\hat{\calA},k}$ is expressed as
\begin{align}
&\bar{\gamma}^{\sf FDD}_{\hat{\calA},k}\nonumber\\
&=\frac{\frac{1-q^2}{M} p_{{\sf DL},k}\mu_{k}^2}{\lambda_k(1+a_k\bar{\mu})^2+\frac{1-q^2}{M}a_k\bar{\lambda}_k\mu_{k}^2+\frac{\bar{\lambda}\sigma_{k}^2}{P_T}(1+a_k\bar{\mu}+\frac{1-q^2}{M}\mu_k)^2},
\end{align}
where $a_k=\frac{(1-q^2)\sigma_{k}^2}{p_{\sf dp}M\tau_d}+\frac{\beta_kq^2}{M}$, $\mu_k=\frac{I\beta_k}{\eta_1/M+\eta_2+\xi}$, $\bar{\mu}=\frac{I}{\eta_1/M+\eta_2+\xi}+\frac{J}{\eta_2+\xi}$,
and $\lambda_k$, $\bar{\lambda}_k$ and $\bar{\lambda}$ can be given by \eqref{lambda_k}, \eqref{bar_lambda_k} and \eqref{bar_lambda}, respectively. Other parameters are given as $\eta_1=\sum_{i=1}^{K}\frac{(1-q^2)\beta_i}{1+e_i}$, $\eta_2=\sum_{i=1}^{K}\frac{a_i}{1+e_i}$, $e'_{i,k}=[(\frac{1}{1-q^2}\qI_K-\qJ)^{-1}\qv_k]_i$, $e'_i=[(\frac{1}{1-q^2}\qI_K-\qJ)^{-1}\qv']_i$, $[\qv_k]_i=\frac{\beta_k(\beta_i/M+a_i/(1-q^2))I}{(\eta_1/M+\eta_2+\xi)^2}$ and
\begin{align}
e_i=&\frac{(\frac{(1-q^2)\beta_i}{M}+a_i)I}{\frac{\eta_1}{M}+\eta_2+\xi}+\frac{a_iJ}{\eta_2+\xi},\\
[\qv']_i=&\frac{(\frac{\beta_i}{M}+\frac{a_i}{1-q^2})I}{(\frac{\eta_1}{M}+\eta_2+\xi)^2}+\frac{\frac{a_i}{1-q^2}J}{(\eta_2+\xi)^2}.
\end{align}
\begin{align}
[\qJ]_{i,j}=&\frac{1-q^2}{(1+e_j)^2}\Bigg(\frac{(\frac{\beta_i}{M}+\frac{a_i}{1-q^2})(\frac{\beta_j}{M}+\frac{a_j}{1-q^2})I}{(\frac{\eta_1}M+\eta_2+\xi)^2} \nonumber \\
 &+\frac{\frac{a_ia_j}{(1-q^2)^2}J}{(\eta_2+\xi)^2} \Bigg).
\end{align}
\end{theorem}}
\begin{proof}
The proof of Theorem \ref{Theorem1FDD} is similar to that of Theorem \ref{Theorem1}, thus we omit it here due to limited space.
\end{proof}

\emph{Remark 5:} From Theorem \ref{Theorem1FDD}, SINR is mainly concerned with downlink pilot length $\tau_d$. For each VR detection, it is impractical to obtain an optimal $\tau_d$ in closed-form since $I$ and $J$ are unknown. To reduce pilot overhead and meet the orthogonality of downlink pilots, $\tau_d$ can be simply set as $I+J$. However, we can figure $\tau_d^{\sf opt}$ for the cases of perfect VR detection and without VR detection, respectively, and the former can be viewed as the upper bound of average {EE}. Invoking Theorem \ref{Theorem1FDD}, $\tau_d^{\sf opt}$ for perfect VR can be obtained by solving
\begin{align}\label{tau_d_opt}
{\max_{\tau_d=L,L+1,\ldots,T-\tau}\quad \frac{W(T-\tau-\tau_d)}{TP_{\sf cir}}\sum_{k=1}^{K}\log(1+\bar{\gamma}_{\calA,k}^{\sf FDD}),}
\end{align}
where $\tau_d^{\sf opt}$ can be easily found by simple one-dimension search.

{\emph{Remark 6:} From Theorem \ref{Theorem1FDD}, the feedback error $q$ will directly impact the SINR on estimated VR. By assuming $q^2=2^{-\frac{B}{I+J-1}}$ \cite{Wagner12TIT}, we observe that the false detection directly impacts the feedback error due to limited feedback bits. Therefore, FDD systems are more sensitive to VR detection error than TDD systems. The missed detection and false detection should be both avoided in FDD systems. }

Based on Theorem \ref{Theorem1FDD}, we can formulate following deterministic problem
{\begin{align}
(\calP3)\ \max_{\zeta_0>0,\tau=K,K+1,\ldots,T-I-J}W&\sum_{I=0}^{L}\sum_{J=0}^{M-L}\Big(\frac{T-\tau-I-J}{TP_{\sf cir}}\nonumber\\
&\times P_{I,J}\bar{R}^{\sf FDD}_{I,J}\Big),
\end{align}
where $\bar{R}^{\sf FDD}_{I,J}=\sum_{k=1}^{K}\log(1+\bar{\gamma}^{\sf FDD}_{\hat{\calA},k})$.} Notice that we have set $\tau_d=I+J$, thus the number of wrongly estimated antennas $J$ makes significant difference to average {EE}. Similar to problem ($\calP2$), ($\calP3$) can also be broken into two subproblems, i.e., optimizing $\zeta_0$ with fixed $\tau$ and optimizing $\tau$ with fixed $\zeta_0$, and eventually solved by alternate optimization. The former subproblem can be solved by proposed detection thresholds in subsection \ref{zeta_opt}, and the latter subproblem can be solved by one-dimension search like solving \eqref{subptau}.

\section{NUMERICAL RESULTS}\label{simulations}
In this section, we validate the accuracy of the derived deterministic approximate results and illustrate the effectiveness of our optimization on $\{\zeta_0,\tau\}$ by numerical simulations. The large-scale factor $\beta$ is modeled as $-10^{-0.53}/d_k^{3.76}$ \cite{Lv18JIoT}, where $d_k$ is measured in meter. For simplicity, we let the downlink noise power and large-scale fading of different user be the same in simulations, i.e., { $\sigma_{1}^2=\sigma_{2}^2=\dots=\sigma_{K}^2$} and $\beta_1=\beta_2=\dots=\beta_K$. We adopt the uniform power allocation and the RZF regularization factor is set as $\xi=\sigma^2_{k}/(p_{{\sf DL},k}MP_T)$ \cite[eq.(28)]{Jun13TWC}. The default settings are as follow:
the power of user pilot $p_{{\sf Tr},k} = 0.1$W ($20$dBm) \cite{GA-20PC}, {the uplink noise power $\sigma_{\sf Tr}^2 = -96$dBm, the bandwidth $W=100$MHz, $P_{\sf syn}=50$mW, $P_{\sf CT}=48.2$mW, $P_{\sf CR}=62.5$mW, $\varsigma=0.35$ \cite{Saba18TWC}, the number of antennas in the BS $M = 128$, the length of VR $L = 16$,} the number of users $K = 4$, the distance between the BS and users $d_k = 400$m, the length of the channel coherence block is $T=200$ symbols and the number of Monte-Carlo trials is $10^4$. For FDD systems, we set $\tau_d=I+J$ to save pilot cost.

\subsection{Accuracy of Deterministic Approximate Results}
We will illustrate the accuracy of our derived deterministic approximate results for $P_{10}$, $P_{01}$, $R_{I,J}^{\sf sum}$ in this part.
\begin{figure}
\includegraphics[width=0.45\textwidth]{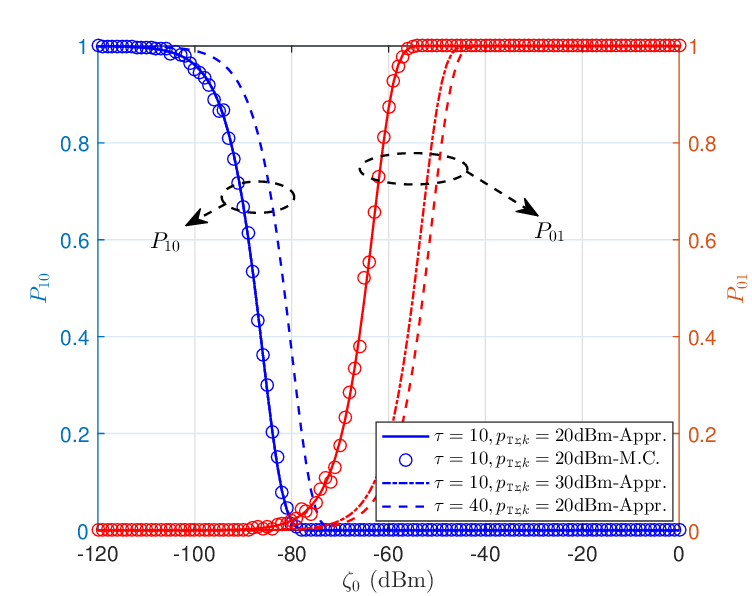}
\centering
\caption{The probabilities of error events versus detection threshold.} \label{Fig1}
\end{figure}
Fig. \ref{Fig1} illustrates the correctness of our derived probabilities of two error events. The suffixes ``M.C." and ``Appr." stand for Monte-Carlo results and the approximate results obtained by \eqref{eq:Pr10-2}, \eqref{eq:Pr01-2}, respectively. The results are identical to each other, which validates the effectiveness of our analysis on VR detection. As the detection threshold rises, the probability of ``false detection", i.e., $P_{10}$, decreases to zero, while the ``missed detection" probability, i.e., $P_{01}$, increases to maximum. To obtain satisfactory average {EE}, the detection threshold needs to be set elaborately. From Fig. \ref{Fig1}, increasing the pilot length or pilot power can both reduce the probabilities of two error events simultaneously. However, the pilot power is budgeted and the pilot length also impacts the average {EE}, which also need to be carefully determined.

\begin{figure}
\includegraphics[width=0.45\textwidth]{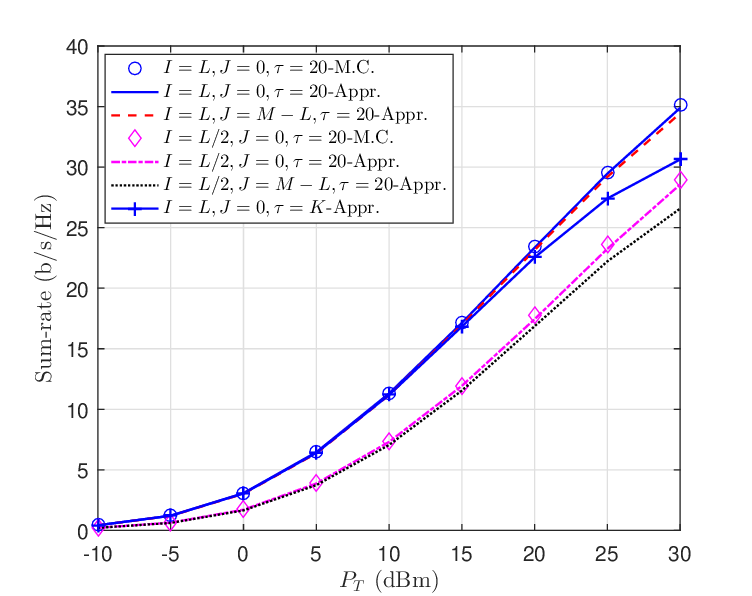}
\centering
\caption{{The sum-rate of different $\{I,J\}$ in TDD systems}.} \label{Fig2}
\end{figure}

Fig. \ref{Fig2} shows the sum-rate performance of different $\{I,J\}$ in TDD mode. First, the deterministic expression of sum-rate can approximate the ergodic result well, even when the number of antennas is limited. {For the curves where $I=L/2$, the performance suffers significantly loss when compared with $I=L$ in the whole power region due to a smaller $I$. The gap caused by false detection is not obvious when $I=L$, whereas this gap grows when $I=L/2$ in high power region. {This is because the linear precoding we adopted utilizes the estimated channel information, where the signal strength on the falsely detected antennas is already very low, and thus the power allocated to the falsely detected antenna is also very low. Therefore, to maintain high sum-rate, it is more important to avoid missed detection.} Besides, the decrease in $\tau$ also leads to performance loss in high power region, because a smaller $\tau$ brings larger channel estimation error.} Therefore, an accurate VR estimation not only helps reduce the precoding complexity but also lift system performances.

\begin{figure}
\includegraphics[width=0.45\textwidth]{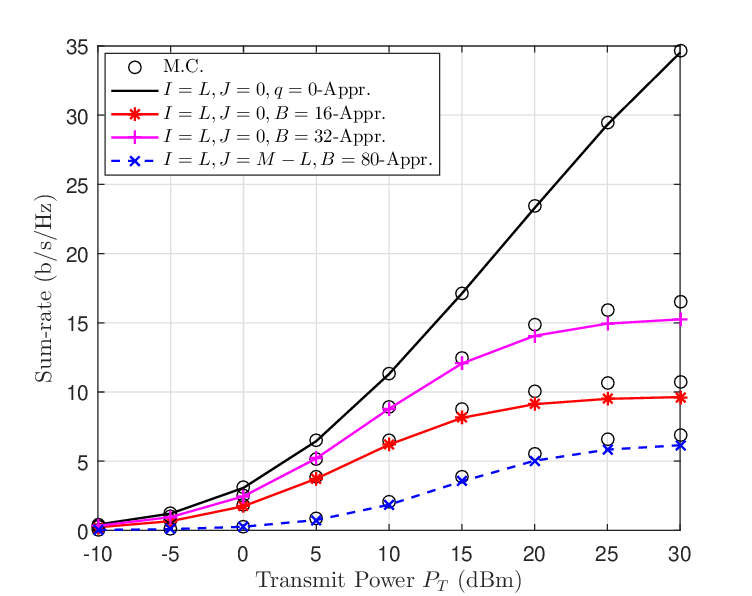}
\centering
\caption{{The sum-rate of different $\{I,J,B\}$ in FDD systems with $\tau_d=20$.}} \label{Fig4}
\end{figure}

{Fig. \ref{Fig4} illustrates the sum-rate of different $\{I,J,B\}$ in FDD system. First, the deterministic approximation can fit the ergodic result well. Compared with TDD systems, the impact of false detection becomes significant. Generally, a larger $B$ lowers $q$ and brings better performance when VR is correctly detected. However, when $J=M-L$, the performance suffers significant loss even if the number of feedback bit is high. Therefore, an accurate VR detection is of great importance for FDD systems, not only helps reduce complexity and channel estimation overhead but also improve performances remarkably.} For clarity, we only show approximate results in the remaining simulations.

\subsection{Optimization on Detection Threshold and Pilot Length}

\begin{figure}
\includegraphics[width=0.45\textwidth]{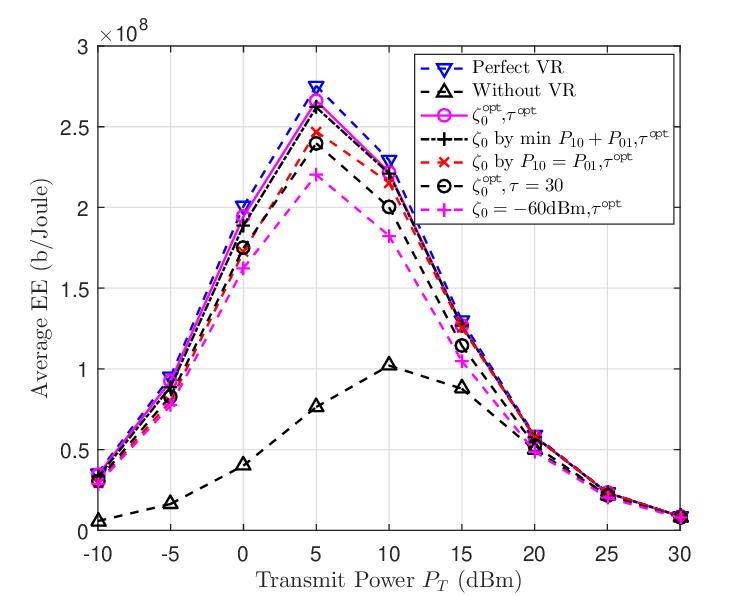}
\centering
\caption{{The comparisons of different detection thresholds in TDD systems.}} \label{Fig5}
\end{figure}

Fig. \ref{Fig5} compares the average {EE} obtained by different detection thresholds with RZF precoder. Notice that the ``Perfect VR" and ``Without VR" curves are obtained by assuming the correct VR estimation, i.e., $I=L,J=0$, and no VR estimation, i.e., $I=L,J=M-L$, respectively. Note that the curve of perfect VR can be viewed as the upper bound here, while the curve of without VR is not the lower bound, which means an improper threshold may result in significant performance loss when compared with no VR estimation condition. {Among the four thresholds shown in Fig. \ref{Fig5}, the optimal threshold $\zeta_0^{\sf opt}$ is obtained by the fminbnd function in MATLAB and it achieves the best performance. The `` $\zeta_0$ by min $P_{10}+P_{01}$" outperforms the threshold obtained by letting $P_{01}=P_{10}$, and it is also very close to the performance obtained by $\zeta_0^{\sf opt}$. The fixed threshold $\zeta_0=-60$dBm results in the poorest {EE}. However, the schemes with VR detection all gain significant improvement when compared with ``Without VR" condition, because the activated antennas consume fixed circuit power even if there is little transmit power allocated on these antennas. This again proves the necessity and importance of accurate VR detection in spatially non-stationary XL-MIMO.}

\begin{figure}
\includegraphics[width=0.45\textwidth]{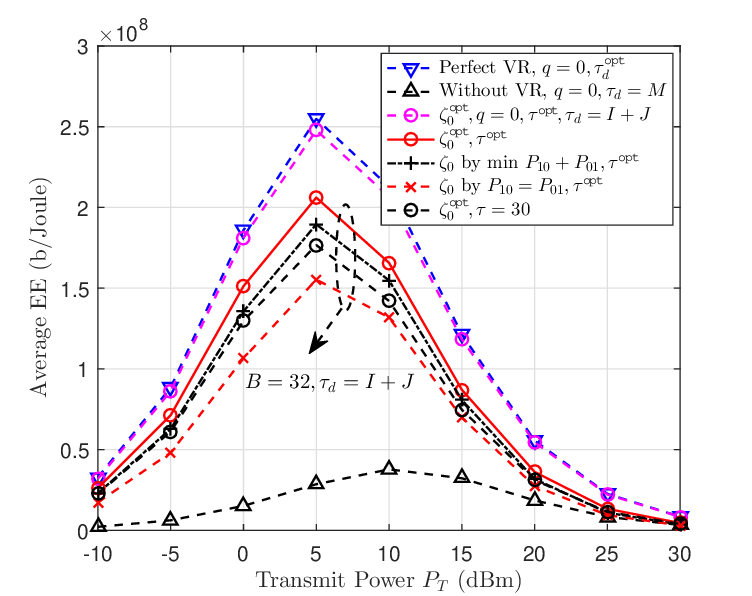}
\centering
\caption{{The average {EE} obtained by alternate optimization and different detection thresholds in FDD systems.}} \label{Fig6}
\end{figure}

Fig. \ref{Fig6} demonstrates the average {EE} obtained by alternate optimization and different detection thresholds in FDD systems, where $p_{\sf dp}=p_{{\sf Tr},k}$. {Obviously, compared with the condition without feedback error, i.e., $q=0$, the EE performance with $B=32$ is heavily affected by the feedback bits. Compared with the result in TDD mode, the performance gap caused by different schemes becomes larger. Besides, we also notice that even without feedback error, the ``Without VR" condition suffers greater performance loss than that in TDD mode because downlink channel estimation in FDD systems requires more pilot. This indicates that FDD systems are more sensitive to VR detection error. In general, the accuracy of VR estimation is crucial for FDD systems, which helps reduce the overhead of training pilot and improve system performances simultaneously.

Based on the results in Fig. \ref{Fig5} and Fig. \ref{Fig6}, we observe that the threshold obtained by minimizing $P_{10}+P_{01}$ can achieve performance close to that by optimal threshold with much lower complexity. Therefore, the detection threshold obtained by minimizing $P_{10}+P_{01}$ is more suitable in practical systems.}

\begin{figure}
\includegraphics[width=0.45\textwidth]{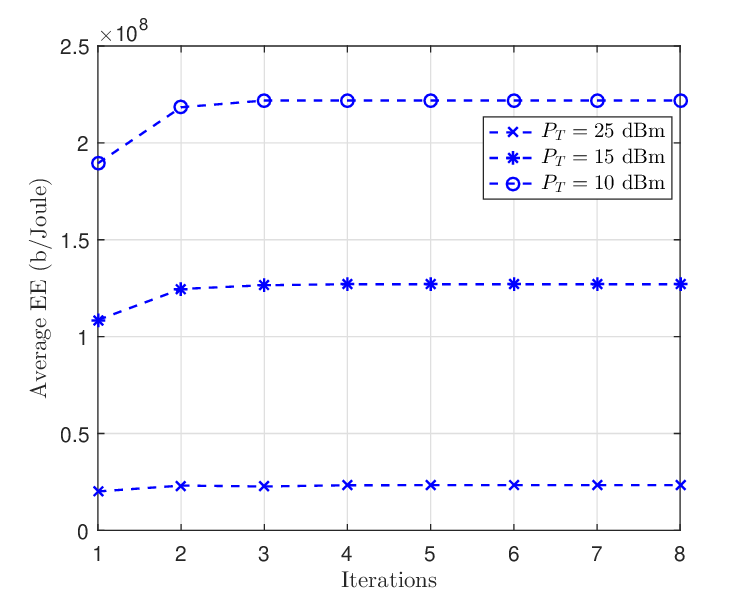}
\centering
\caption{{Convergence behaviour of alternate optimization in TDD systems.}} \label{iterations}
\end{figure}

As illustrated in Fig. \ref{iterations}, the proposed alternate optimization algorithm can quickly converge to a fixed point with different transmit power. Notice that the convergence may not be monotonically increasing with more iterations due to the initial settings and discrete value of $\tau$.


\begin{figure}
\includegraphics[width=0.45\textwidth]{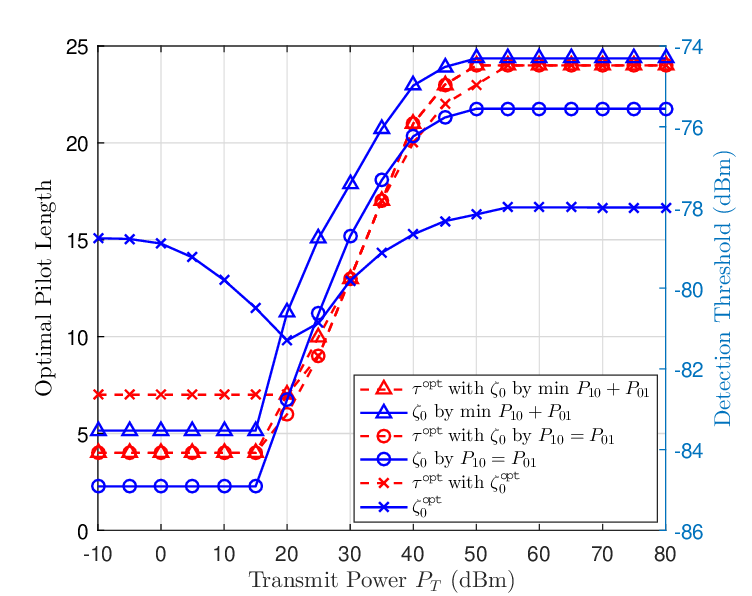}
\centering
\caption{{The optimal pilot length and detection threshold by different schemes in TDD systems.}} \label{Fig7}
\end{figure}

Fig. \ref{Fig7} illustrates the optimal pilot length and detection threshold by different schemes. {Generally, both the optimal pilot length and detection threshold increase as transmit power grows from $20$dBm to $50$dBm, and eventually saturate at fixed values with sufficiently high transmit power. In the low power region, e.g., -10dBm to 20dBm, the threshold by minimizing $P_{10}+P_{01}$ is closer to the optimal result, {and the pilot length by minimizing $P_{10}+P_{01}$ is almost the same as that by $P_{10}=P_{01}$}. This indicates that in low power region, EE with imperfect VR is more sensitive to the setting of threshold than optimizing the pilot length. }

\section{CONCLUSION}\label{conclusions}
This paper investigates the ergodic downlink sum-rate and average {EE} by considering the VR estimation error in a multi-user XL-MIMO system. Exploiting the energy detection method, we derive the probabilities of two error events, i.e., the missed detection and false detection, during the VR estimation stage. To reduce computation complexity, we perform uplink channel estimation on detected VR in TDD mode, or send downlink pilot for channel estimation in FDD mode. Based on estimated channels, the ergodic average {EE} maximization problem is formulated by adopting RZF precoder. To avoid massive Monte-Carlo trials, we derive deterministic approximate expressions for average {EE} in both TDD and FDD systems to optimize uplink pilot length and detection threshold. The simulations reveal the influences of VR estimation error and demonstrate the effectiveness of our proposed schemes.

\begin{appendices}
\section{Proof of Theorem \ref{Theorem1}}\label{The1proof}
Before giving detailed proof on Theorem \ref{Theorem1}, we first introduce some useful lemmas and theorems.
{\begin{lemma}\label{lemma1}[Matrix inversion lemma]\emph{(\cite[lemma 2.2]{Silverstein95JMA})}
Let $\qC\in\mathbb{C}^{M\times M}$ be a Hermitian invertible matrix and for any $\qa\in\mathbb{C}^{M\times 1},c\in\mathbb{C}$ that $\qC+c\qa\qa^{H}$ is invertible, we have
\begin{align}
\qa^{H}(\qC+c\qa\qa^{H})^{-1}=\frac{\qa^{H}\qC^{-1}}{1+c\qa^{H}\qC^{-1}\qa}.
\end{align}
\end{lemma}

\begin{lemma}\label{2-3lem}\emph{\cite[lemma 14.2]{Couillet11}, \cite[lemma 5]{Wagner12TIT}, \cite[lemma 4]{Hoydis13JSAC}:}
Assume that $\qC\in\mathbb{C}^{M\times M}$ has uniformly bounded spectral norm with respect to $M$, and for any mutually independent (also independent to $\qC$) $\qa\in\mathbb{C}^{M\times1},\qb\in\mathbb{C}^{M\times1}\sim\mathcal{CN}(\mathbf{0},\frac{1}{M}\qI_{M})$, as $M\rightarrow\infty$, we have
$(i).\ \qa^{H}\qC\qa-\frac{1}{M}{\sf tr}(\qC)\stackrel{\sf a.s.}{\longrightarrow}0$,
$(ii).\ \qa^{H}\qC\qb\stackrel{\sf a.s.}{\longrightarrow}0$,
$(iii).\ \mathbb{E}\{|\qa^{H}\qC\qa|^2\}-\mathbb{E}\{|\frac{1}{M}{\sf tr}(\qC)|^2\}\stackrel{\sf a.s.}{\longrightarrow}0$.
\end{lemma}

\begin{lemma}\label{lemma3}\emph{\cite[lemma 14.3]{Couillet11}:}
Let $\qB\in\mathbb{C}^{M\times M}$ be a Hermitian nonnegative definite matrix, and for any $\qC$ in Lemma \ref{2-3lem} and $\qa\in\mathbb{C}^{M\times1}$, as $M\rightarrow\infty$, we have
\begin{align}
\frac{1}{M}{\sf tr}(\qC\qB^{-1})-\frac{1}{M}{\sf tr}\Big(\qC(\qB+\qa\qa^{H})^{-1}\Big)\stackrel{\sf a.s.}{\longrightarrow}0.
\end{align}
\end{lemma}

Notice that the condition $M\rightarrow\infty$ in Lemma \ref{2-3lem} and Lemma \ref{lemma3} can be satisfied when $\mathcal{N}\rightarrow\infty$.

\begin{theorem}\label{Theorem2}\emph{\cite[theorem 1]{Wagner12TIT}:}
Let $\qQ\in\mathbb{C}^{M\times M}$ and $\qT_{k}\in\mathbb{R}^{M\times M}(k=1,2,\dots,K)$ be deterministic Hermitian nonnegative definite matrix with uniformly bounded spectral norm with respect to $M$. For $c>0$ and a random matrix $\qF=[\qf_{1},\qf_{2},\ldots,\qf_{K}]\in\mathbb{C}^{M\times K}$ which satisfies $\qf_{k}\sim\mathcal{CN}(\mathbf{0},\frac{1}{M}\qT_{k})$, as $\calN\rightarrow\infty$, we have
\begin{align}
\frac{1}{M}{\sf tr}\Big(\qQ(\qF\qF^{H}+c\qI_{M})^{-1}\Big)-\frac{1}{M}{\sf tr}(\qQ\qS)\stackrel{\sf a.s.}{\longrightarrow}0,
\end{align}
where
\begin{align}
\qS=\Big(\frac{1}{M}\sum_{k=1}^{K}\frac{\qT_k}{1+e_k}+c\qI_{M}\Big)^{-1},
\end{align}
and $e_k=\frac{1}{M}{\sf tr}(\qT_k\qS)$ for $k=1,2,\ldots,K$.
\end{theorem}}

Invoking above lemmas and Theorem \ref{Theorem2}, we aim to derive an approximate SINR instead of performing Monte-Carlo averaging, which is of high computation complexity.

Note that the ergodic SINR \eqref{ergSINR} is composed of three part with expectation: (\emph{A}). the signal power $p_{{\sf DL},k}|\bbE\{\qh_{\hat{\calA},k}^{H}\qw_k\}|^2$; (\emph{B}). the interference power $\sum_{j\neq k}^{K}p_{{\sf DL},j}\bbE\{|\qh_{\hat{\calA},k}^{H}\qw_j|^2\}$; (\emph{C}). the normalized factor $\alpha$. The ergodic SINR $\gamma_{k}$ for $k=1,2,\ldots,K$ will be approximated part by part. For ease of notation, we simply denote $\hat{\qh}_{\hat{\calA},k}$, $\qh_{\hat{\calA},k}$, $\hat{\qH}_{\hat{\calA},{\sf LS}}$ and $\qH_{\hat{\calA}}$ as $\hat{\qh}_{k}$, $\qh_{k}$, $\hat{\qH}$ and $\qH$, respectively, where $\qh_k=\qR_k^{\frac{1}{2}}\qz_k$, $\qR_k=[\qR_k^{\sf full}]_{\hat{\calA}}=\beta_k\mathbf{\Lambda}$ and $\qz_k=[\qz_k^{\sf full}]_{\hat{\calA}}$.

\subsection{Deterministic approximation for signal power}
The deterministic approximation for $\bbE\{\qh_{k}^{H}\qw_k\}$ is given by Lemma \ref{lemma-signalpower}.
\begin{lemma}\label{lemma-signalpower}
As $\calN\rightarrow\infty$, we have
\begin{align}
\qh_{k}^{H}\qw_k-\frac{\alpha\frac{1}{M}\mu_k}{1+\frac{1}{M}\mu_k+\frac{\sigma_{\sf Tr}^2}{M\tau p_{{\sf Tr},k}}\bar{\mu}}\stackrel{\sf a.s.}{\longrightarrow}0,
\end{align}
where $\mu_k$ and $\bar{\mu}$ are given in Theorem \ref{Theorem1}.
\end{lemma}
\begin{proof}
With RZF precoding, we have $\frac{1}{\alpha}\qh_{k}^{H}\qw_k=\qh_{k}^{H}(\hat{\qH}\hat{\qH}^{H}+\xi\qI_{I+J})^{-1}\hat{\qh}_{k}$. Similar to \cite[eq.(41)]{Jun13TWC}, applying Lemma \ref{lemma1} and \ref{2-3lem}, we have $\qh_{k}^{H}(\qA_{[k]}+\hat{\qh}_{k}\hat{\qh}_{k}^{H})^{-1}\hat{\qh}_{k}\stackrel{\sf a.s.}{\longrightarrow}\frac{{\sf tr}(\qR_k\qA_{[k]}^{-1})/M}{1+\hat{\qh}_{k}^{H}\qA_{[k]}^{-1}\hat{\qh}_{k}}$,
where $\qA_{[k]}=\sum_{i\neq k}^{K}\hat{\qh}_{i}\hat{\qh}_{i}^{H}+\xi\qI_{I+J}$.
Similarly, using Lemma \ref{2-3lem} again, we have
\begin{align}
\hat{\qh}_{k}^{H}\qA_{[k]}^{-1}\hat{\qh}_{k}\stackrel{\sf a.s.}{\longrightarrow}\frac{1}{M}{\sf tr}(\qR_k\qA_{[k]}^{-1})+\frac{\sigma_{\sf Tr}^2}{M\tau p_{{\sf Tr},k}}{\sf tr}(\qA_{[k]}^{-1}).\label{hathAhath}
\end{align}
Then applying Lemma \ref{lemma3}, as $\calN\rightarrow\infty$, we have $\frac{1}{M}{\sf tr}(\qR_k\qA_{[k]}^{-1})\stackrel{\sf a.s.}{\longrightarrow}\frac{1}{M}{\sf tr}(\qR_k\qA^{-1})$ and $\frac{1}{M}{\sf tr}(\qA_{[k]}^{-1})\stackrel{\sf a.s.}{\longrightarrow}\frac{1}{M}{\sf tr}(\qA^{-1})$. Further applying Theorem \ref{Theorem2}, we have $\frac{1}{M}{\sf tr}(\qR_k\qA^{-1})\stackrel{\sf a.s.}{\longrightarrow}\frac{1}{M}\mu_{k}$ and $\frac{1}{M}{\sf tr}(\qA^{-1})\stackrel{\sf a.s.}{\longrightarrow}\frac{1}{M}\bar{\mu}$,
where $\mu_{k}={\sf tr}(\qR_k\qT), \bar{\mu}={\sf tr}(\qT)$, $\bar{\qR}_{j}=\frac{\beta_{j}}{M}(\mathbf{\Lambda}+\frac{\sigma_{\sf Tr}^2}{\tau g}\qI_{I+J})$, $\qT=\Big(\frac{\eta}{M}\mathbf{\Lambda}+(\frac{\eta\sigma^2_{\sf Tr}}{Mg\tau}+\xi)\qI_{I+J}\Big)^{-1}$ and $e_j={\sf tr}(\bar{\qR}_{j}\qT)$.
Both $\bar{\qR}_{j}$ and $\qT$ are diagonal, by substituting $\bar{\qR}_{j}$ and $\qT$ into above equations, we have $\mu_{k}$, $\bar{\mu}$ and $e_j$ as shown in \eqref{mu_k}, \eqref{bar_mu} and \eqref{e}, respectively. Thus Lemma \ref{lemma-signalpower} is proved.
\end{proof}

\subsection{Deterministic approximation for interference power}
With Lemma \ref{lemma-signalpower}, the signal power term is approximated by deterministic expression. Then, the interference power $\sum_{j\neq k}^{K}p_{{\sf DL},j}\bbE\{|\qh_{k}^{H}\qw_j|^2\}$ can be approximated by following lemma.
\begin{lemma}\label{lemma5}
As $\calN\rightarrow\infty$, we have
\begin{align}
\sum_{j\neq k}^{K}p_{{\sf DL},j}\bbE\{|\qh_{k}^{H}\qw_j|^2\}\stackrel{\sf a.s.}{\longrightarrow}\frac{\frac{\lambda_{k}}{M}(1+a_k\bar{\mu})^2+a_k\bar{\lambda}_{k}\mu_{k}^2}{\alpha^2(1+a_k\bar{\mu}+\frac{d^2}{M}\mu_k)^2},
\end{align}
where $\lambda_k$, $a_k$, $\bar{\mu}$, $\bar{\lambda}_k$ and $\mu_k$ are given in Theorem \ref{Theorem1}.
\end{lemma}
\begin{proof}
We first rewrite the interference by
\begin{align}
\sum_{j\neq k}^{K}\frac{p_{{\sf DL},j}}{\alpha^2}|\qh_{k}^{H}\qw_j|^2=\qh_{k}^{H}\qA^{-1}\hat{\qH}_{[k]}\qP_{[k]}\hat{\qH}^{H}_{[k]}\qA^{-1}\qh_{k}.
\end{align}
The right hand side of above equation can be further rewritten as
\begin{align}
 &\qh_{k}^{H}\qA^{-1}\hat{\qH}_{[k]}\qP_{[k]}\hat{\qH}^{H}_{[k]}\qA^{-1}\qh_{k} \nonumber \\
\stackrel{(a)}{=}&\qh_{k}^{H}\qA_{[k]}^{-1}\hat{\qH}_{[k]}\qP_{[k]}\hat{\qH}^{H}_{[k]}\qA^{-1}\qh_{k}\nonumber\\
&+\qh_{k}^{H}\qA^{-1}(\qA_{[k]}-\qA)\qA_{[k]}^{-1}\hat{\qH}_{[k]}\qP_{[k]}\hat{\qH}^{H}_{[k]}\qA^{-1}\qh_{k},\label{mideq1}
\end{align}
where $(a)$ exploits the fact that $\qA^{-1}-\qA^{-1}_{[k]}=\qA^{-1}(\qA_{[k]}-\qA)\qA_{[k]}^{-1}$. Then we have
\begin{align}
\qA-\qA_{[k]}=&\qR_k^{\frac{1}{2}}\qz_k\qz_k^H\qR_{k}^{\frac{1}{2}} + \sqrt{a_k}\bar{\qe}_k\qz_{k}^{H}\qR_{k}^{\frac{1}{2}}\nonumber \\
&+\sqrt{a_k}\qR_{k}^{\frac{1}{2}}\qz_k\bar{\qe}_k^{H}+a_k\bar{\qe}_{k}\bar{\qe}_{k}^{H},\label{mideq2}
\end{align}
where $\bar{\qe}_k=\frac{1}{\sqrt{a_k}}\qe_k$ with respect to $\mathcal{CN}(\mathbf{0},\qI_{I+J})$.
Substituting above equation into \eqref{mideq1}, we have
\begin{align}\label{total_interfer}
 &\qh_{k}^{H}\qA_{[k]}^{-1}\hat{\qH}_{[k]}\qP_{[k]}\hat{\qH}^{H}_{[k]}\qA^{-1}\qh_{k} \nonumber\\
 &+\qh_{k}^{H}\qA^{-1}(\qA_{[k]}-\qA)\qA_{[k]}^{-1}\hat{\qH}_{[k]}\qP_{[k]}\hat{\qH}^{H}_{[k]}\qA^{-1}\qh_{k}\nonumber\\
=&(1-\qz_k^{H}\qC_k\qz_k-\sqrt{a_k}\qz_k^{H}\qC'_k\bar{\qe}_k)\qz_k^{H}\qB_k\qz_k \nonumber\\
 &-(\sqrt{a_k}\qz_k^{H}\qC_k\qz_k+a_k\qz_k^{H}\qC'_k\bar{\qe}_k)\bar{\qe}_k^{H}\qB'_k\qz_k,
\end{align}
where $\qC_k=\qR_{k}^{\frac{1}{2}}\qA^{-1}\qR_{k}^{\frac{1}{2}}$, $\qC'_k=\qR_{k}^{\frac{1}{2}}\qA^{-1}$, $\qB_k$ and $\qB'_k$ are respectively given by
\begin{align}
\qB_k&=\qR_{k}^{\frac{1}{2}}\qA_{[k]}^{-1}\hat{\qH}_{[k]}\qP_{[k]}\hat{\qH}^{H}_{[k]}\qA^{-1}\qR_{k}^{\frac{1}{2}},\nonumber\\
\qB'_k&=\qA_{[k]}^{-1}\hat{\qH}_{[k]}\qP_{[k]}\hat{\qH}^{H}_{[k]}\qA^{-1}\qR_{k}^{\frac{1}{2}},\nonumber
\end{align}
By further denoting $\qU_k=\qA_{[k]}^{-1}\hat{\qH}_{[k]}\qP_{[k]}\hat{\qH}_{[k]}^{H}$, we have $\qz_{k}^{H}\qB_k\qz_k=\qh_k^{H}\qU_k\qA^{-1}\qh_k$.
Notice that
\begin{align}
\qh_k^{H}\qU_k\qA^{-1}\qh_k-\qh_k^{H}\qU_k\qA_{[k]}^{-1}\qh_k=\qh_k^{H}\qU_k(\qA^{-1}-\qA_{[k]}^{-1})\qh_k.\label{mideq4}
\end{align}
Substituting $\qA^{-1}-\qA_{[k]}^{-1}=\qA^{-1}(\qA_{[k]}-\qA)\qA_{[k]}^{-1}$ and \eqref{mideq2} into \eqref{mideq4}, we have
\begin{align}
 &\qh_k^{H}\qU_k\qA_{[k]}^{-1}\qh_k-\qh_k^{H}\qU_k \qA^{-1}\qh_k  \nonumber\\
=&\qh_k^{H}\qU_k\qA^{-1}\qh_k\qh_k^{H}\qA_{[k]}^{-1}\qh_k \nonumber\\ &+\sqrt{a_k}\qh_k^{H}\qU_k\qA^{-1}\bar{\qe}_{k}\qh_{k}^{H}\qA_{[k]}^{-1}\qh_k\nonumber\\
&+\sqrt{a_k}\qh_k^{H}\qU_k\qA^{-1}\qh_k\bar{\qe}_k^{H}\qA_{[k]}^{-1}\qh_k \nonumber\\
&+a_k\qh_k^{H}\qU_k\qA^{-1}\bar{\qe}_{k}\bar{\qe}_{k}^{H}\qA_{[k]}^{-1}\qh_k.\label{mideq5}
\end{align}
The above equation can be reformulated as
\begin{align}\label{mideq3}
 &\qh_{k}^{H}\qU_k\qA^{-1}\qh_k\Big(1+\qh_{k}^{H}\qA_{[k]}^{-1}\qh_k+\sqrt{a_k}\bar{\qe}_{k}^{H}\qA_{[k]}^{-1}\qh_k\Big)\nonumber\\
=&\qh_{k}^{H}\qU_k\qA_{[k]}^{-1}\qh_{k}-\qh_{k}^{H}\qU_k\qA^{-1}\bar{\qe}_k \nonumber\\
&\times\Big(\sqrt{a_k}\qh_{k}^{H}\qA_{[k]}^{-1}\qh_k+a_k\bar{\qe}_{k}^{H}\qA_{[k]}^{-1}\qh_k\Big).
\end{align}
Using Lemma \ref{2-3lem} $(i)$ and $(ii)$, we have $\bar{\qe}_{k}^{H}\qA_{[k]}^{-1}\qh_k\stackrel{\sf a.s.}{\longrightarrow}0$,  $\qh_{k}^{H}\qA_{[k]}^{-1}\qh_k\stackrel{\sf a.s.}{\longrightarrow}\frac{1}{M}{\sf tr}(\qR_k\qA_{[k]}^{-1})$ and $\qh_{k}^{H}\qU_k\qA_{[k]}^{-1}\qh_{k}\stackrel{\sf a.s.}{\longrightarrow}\frac{1}{M}{\sf tr}(\qR_k\qU_k\qA_{[k]}^{-1})$, thus \eqref{mideq3} can be rewritten as
\begin{align}
&(1+\frac{1}{M}\mu_k)\qh_{k}^{H}\qU_k\qA^{-1}\qh_k \nonumber\\
=&\frac{1}{M}{\sf tr}(\qR_k\qU_k\qA_{[k]}^{-1}) -\frac{\sqrt{a_k}}{M}\mu_k\qh_{k}^{H}\qU_k\qA^{-1}\bar{\qe}_k.\label{hUVequation1}
\end{align}
Apart from the above equation, similar to \eqref{mideq4}, we also have $\qh_k^{H}\qU_k\qA^{-1}\bar{\qe}_k-\qh_k^{H}\qU_k\qA_{[k]}^{-1}\bar{\qe}_k=\qh_k^{H}\qU_k(\qA^{-1}-\qA_{[k]}^{-1})\bar{\qe}_k$
Then, similar to \eqref{mideq5}-\eqref{mideq3}, we can rewrite the above equation as
\begin{align}
 &\qh_{k}^{H}\qU_k\qA^{-1}\bar{\qe}_{k}\Big(1+a_k\bar{\qe}_{k}^{H}\qA_{[k]}^{-1}\bar{\qe}_{k}+\sqrt{a_k}\qh_{k}^{H}\qA_{[k]}^{-1}\bar{\qe}_{k}\Big)\nonumber\\
=&\qh_{k}^{H}\qU_k\qA_{[k]}^{-1}\bar{\qe}_{k}-\qh_{k}^{H}\qU_k\qA^{-1}\qh_k \nonumber\\
&\times \Big(\sqrt{a_k}\bar{\qe}_{k}^{H}\qA_{[k]}^{-1}\bar{\qe}_{k}+\qh_{k}^{H}\qA_{[k]}^{-1}\bar{\qe}_{k}\Big).
\end{align}
Again using Lemma \ref{2-3lem} $(i)$ and $(ii)$, we have $\qh_{k}^{H}\qU_k\qA_{[k]}^{-1}\bar{\qe}_{k}\stackrel{\sf a.s.}{\longrightarrow}0$ and $\frac{1}{M}\bar{\qe}_{k}^{H}\qA_{[k]}^{-1}\bar{\qe}_{k}\stackrel{\sf a.s.}{\longrightarrow}\frac{1}{M}{\sf tr}(\qA_{[k]}^{-1})$. Since ${\sf tr}(\qA_{[k]}^{-1})$ is approximated by $\bar{\mu}$, we have $\qh_{k}^{H}\qU_k\qA^{-1}\bar{\qe}_{k}=\frac{1+a_k\bar{\mu}}{-\sqrt{a_k}\bar{\mu}}\qh_{k}^{H}\qU_k\qA^{-1}\qh_k$, which can be together solved with \eqref{hUVequation1} and then we obtain
\begin{align}
\qh_{k}^{H}\qU_k\qA^{-1}\qh_k&=\frac{\frac{1}{M}(1+a_k\bar{\mu}){\sf tr}(\qR_k\qU_k\qA_{[k]}^{-1})}{1+a_k\bar{\mu}+\frac{1}{M}\mu_k},\label{zBz}\\
\qh_{k}^{H}\qU_k\qA^{-1}\bar{\qe}_{k}&=\frac{-\frac{\sqrt{a_k}}{M}\bar{\mu}{\sf tr}(\qR_k\qU_k\qA_{[k]}^{-1})}{1+a_k\bar{\mu}+\frac{1}{M}\mu_k}.
\end{align}
Therefore, let $\qU_k=\qI_{I+J}$ and we have
\begin{align}
\qz_k^{H}\qC_k\qz_k=\qh_{k}^{H}\qA^{-1}\qh_k&=\frac{\frac{1}{M}\mu_k(1+a_k\bar{\mu})}{1+a_k\bar{\mu}+\frac{1}{M}\mu_k},\label{zCz}\\
\qz_k^{H}\qC'_k\bar{\qe}_k=\qh_{k}^{H}\qA^{-1}\bar{\qe}_k&=\frac{-\frac{1}{M}\sqrt{a_k}\mu_k\bar{\mu}}{1+a_k\bar{\mu}+\frac{1}{M}\mu_k}.\label{zC'e}
\end{align}
Now we have obtained deterministic approximations for $\qz_k^{H}\qB_k\qz_k$, $\qz_k^{H}\qC_k\qz_k$, $\qz_k^{H}\qC'_k\bar{\qe}_k$, respectively. For $\bar{\qe}_k^{H}\qB'_k\qz_k=\bar{\qe}_k^{H}\qU_k\qA^{-1}\qh_k$, we have following equation similar to \eqref{mideq4}
\begin{align}
\bar{\qe}_k^{H}\qU_k\qA^{-1}\qh_k-\bar{\qe}_k^{H}\qU_k\qA_{[k]}^{-1}\qh_k=\bar{\qe}_k^{H}\qU_k(\qA^{-1}-\qA_{[k]}^{-1})\qh_k.
\end{align}
Substituting \eqref{mideq2} into above equation and using previous approximations, we have
\begin{align}
\bar{\qe}_k^{H}\qU_k\qA^{-1}\qh_k(1+\frac{1}{M}\mu_k)=-\frac{\sqrt{a_k}}{M}\mu_k\bar{\qe}_k^{H}\qU_k\qA^{-1}\bar{\qe}_k.\label{mideq7}
\end{align}
Similarly, for the term $\bar{\qe}_k^{H}\qU_k\qA^{-1}\bar{\qe}_k$, we also have following equation by utilizing Lemma \ref{2-3lem} $(i)$,
\begin{align}
 &\bar{\qe}_k^{H}\qU_k\qA^{-1}\bar{\qe}_k(1+a_k\bar{\mu}) \nonumber \\
=&{\sf tr}(\qU_k\qA_{[k]}^{-1}) -\sqrt{a_k}\bar{\mu}\bar{\qe}_k^{H}\qU_k\qA^{-1}\qh_k, \label{mideq8}
\end{align}
By together solving \eqref{mideq7} and \eqref{mideq8}, we have
\begin{align}\label{eB'z}
\bar{\qe}_k^{H}\qB'_k\qz_k=\bar{\qe}_k^{H}\qU_k\qA^{-1}\qh_k=\frac{-\frac{\sqrt{a_k}}{M}\mu_k{\sf tr}(\qU_k\qA^{-1}_{[k]})}{1+a_k\bar{\mu}+\frac{1}{M}\mu_k}.
\end{align}
Using Lemma \ref{lemma3}, we have ${\sf tr}(\qR_k\qU_k\qA_{[k]}^{-1})\stackrel{\sf a.s.}{\longrightarrow}{\sf tr}(\qP_{[k]}\hat{\qH}_{[k]}^{H}\qA^{-1}\qR_k\qA^{-1}\hat{\qH}_{[k]})$, which can be further approximated by following lemma.
\begin{lemma}\label{lemma4}
As $\calN\rightarrow\infty$, we have $\frac{1}{M}{\sf tr}(\qP_{[k]}\hat{\qH}_{[k]}^{H}\qA^{-1}\qR_k\qA^{-1}\hat{\qH}_{[k]})\stackrel{\sf a.s.}{\longrightarrow}\frac{1}{M}\lambda_k$,
where $\lambda_k$ is given by \eqref{lambda_k}.
\end{lemma}
\begin{proof}
The proof of Lemma \ref{lemma4} is similar to \cite[eq.(125)-(129)]{Wagner12TIT}, thus omitted here.
\end{proof}
Let $\qR_k=\qI_{I+J}$ in Lemma \ref{lemma4}, we have $\frac{1}{M}{\sf tr}(\qP_{[k]}\hat{\qH}_{[k]}^{H}\qA^{-2}\hat{\qH}_{[k]})\stackrel{\sf a.s.}{\longrightarrow}\frac{1}{M}\bar{\lambda}_k$
The derivation is similar to the proof of Lemma \ref{lemma4}. With Lemma \eqref{lemma4}, we have $\bar{\qe}_k^H\qB'_k\qz_k=\frac{-\sqrt{a_k}\mu_k\bar{\lambda}_k/M}{1+a_k\bar{\mu}+\mu_k/M}$
Finally, substituting the terms $\qz_k^{H}\qB_k\qz_k$, $\qz_k^{H}\qC_k\qz_k$, $\qz_k^{H}\qC'_k\bar{\qe}_k$ and $\bar{\qe}_k^H\qB'_k\qz_k$ in \eqref{total_interfer} with deterministic approximations \eqref{zBz}, \eqref{zCz}, \eqref{zC'e} and \eqref{eB'z}, respectively, we have Lemma \ref{lemma5} proved.
\end{proof}

\subsection{Deterministic approximation for normalized factor}
With abovementioned notations, we have normalized factor by $\alpha=\sqrt{\frac{MP_{T}}{\mathbb{E}\{{\sf tr}(\qP_{\sf DL}\hat{\qH}^{H}\qA^{-2}\hat{\qH})\}}}$.
We focus on the term ${\sf tr}(\qP_{\sf DL}\hat{\qH}^{H}\qA^{-2}\hat{\qH})$, which is similar to the term ${\sf tr}(\qP_{[k]}\hat{\qH}_{[k]}^{H}\qA^{-2}\hat{\qH}_{[k]})$. By letting $\qR_k=\qI_{I+J}$ in the proof of Lemma \ref{lemma4}, we can easily obtain the approximation $\frac{1}{M}{\sf tr}(\qP_{\sf DL}\hat{\qH}^{H}\qA^{-2}\hat{\qH})-\frac{1}{M}\bar{\lambda}\stackrel{\sf a.s.}{\longrightarrow}0$ as $\calN\rightarrow\infty$, where $\bar{\lambda}$ is given in \eqref{bar_lambda}.

By substituting all approximations into the ergodic SINR, we have Theorem \ref{Theorem1} proved.

\end{appendices}

\end{document}